\documentclass[pdftex,11pt,a4paper]{article}

\usepackage[percent]{overpic}
\usepackage{mathtools} % needed for \coloneqq (:=)
\usepackage{amsthm}

\newcommand\independent{\protect\mathpalette{\protect\independenT}{\perp}}
\def\independenT#1#2{\mathrel{\rlap{$#1#2$}\mkern2mu{#1#2}}}

 \newcounter{thm}
 \newcounter{ex}
 \newcounter{re}
 \newtheorem{Proposition}[thm]{Proposition}

\usepackage[textsize=tiny]{todonotes}

\usepackage{color}

\usepackage{subcaption}
\usepackage{authblk}
\usepackage[top=1in, bottom=2in, left=1.25in, right=1.25in]{geometry}

\begin{document}

% Full title of the paper (Capitalized)
\title{Invariant components of synergy, redundancy, and unique information among three variables}

\date{\vspace{-5ex}}

% Authors, for the paper (add full first names)
\author[1]{Giuseppe Pica \thanks{giuseppe.pica@iit.it}}
\author[1]{Eugenio Piasini}
\author[1,2]{Daniel Chicharro}
\author[1]{Stefano Panzeri \thanks{stefano.panzeri@iit.it}}
\affil[1]{Neural Computation Laboratory, Center for Neuroscience and Cognitive Systems @UniTn, Istituto Italiano di Tecnologia, Rovereto (TN) 38068, Italy}
\affil[2]{Department of Neurobiology, Harvard Medical School, Boston, MA 02115, USA}

\maketitle

% Simple summary
%\simplesumm{}

% Abstract (Do not use inserted blank lines, i.e. \\)
\begin{abstract}%A single paragraph of about 200 words maximum. For research articles, abstracts should give a pertinent overview of the work. We strongly encourage authors to use the following style of structured abstracts, but without headings: 1) Background: Place the question addressed in a broad context and highlight the purpose of the study; 2) Methods: Describe briefly the main methods or treatments applied; 3) Results: Summarize the article's main findings; and 4) Conclusion: Indicate the main conclusions or interpretations. The abstract should be an objective representation of the article: it must not contain results which are not presented and substantiated in the main text and should not exaggerate the main conclusions.
In a system of three stochastic variables, the Partial Information Decomposition (PID) of Williams and Beer dissects the information that two variables (sources) carry about a third variable (target) into nonnegative information atoms that describe redundant, unique, and synergistic modes of dependencies among the variables. However, the classification of the three variables into two sources and one target limits the dependency modes that can be quantitatively resolved, and does not naturally suit all systems. 
Here, we extend the PID to describe trivariate modes of dependencies in full generality, without introducing additional decomposition axioms or making assumptions about the target/source nature of the variables. By comparing different PID lattices of the same system, we unveil a finer PID structure made of seven nonnegative information subatoms that are invariant to different target/source classifications and that are sufficient to construct any PID lattice. This finer structure naturally splits redundant information into two nonnegative components: the source redundancy, which arises from the pairwise correlations between the source variables, and the non-source redundancy, which does not, and relates to the synergistic information the sources carry about the target. The invariant structure is also sufficient to construct the system's entropy, hence it characterizes completely all the interdependencies in the system.
\end{abstract}

% The fields PACS, MSC, and JEL may be left empty or commented out if not applicable
%\PACS{J0101}
%\MSC{}
%\JEL{}

%%%%%%%%%%%%%%%%%%%%%%%%%%%%%%%%%%%%%%%%%%
% Only for the journal Data:

%\dataset{DOI number or link to the deposited data set in cases where the data set is published or set to be published separately. If the data set is submitted and will be published as a supplement to this paper in the journal Data, this field will be filled by the editors of the journal. In this case, please make sure to submit the data set as a supplement when entering your manuscript into our manuscript editorial system.}

%\datasetlicense{license under which the data set is made available (CC0, CC-BY, CC-BY-SA, CC-BY-NC, etc.)}

%%%%%%%%%%%%%%%%%%%%%%%%%%%%%%%%%%%%%%%%%%

%%%%%%%%%%%%%%%%%%%%%%%%%%%%%%%%%%%%%%%%%%
%% Sections that are not mandatory are listed as such. The section titles given are for Articles. Review papers and other article types have a more flexible structure.

%%%%%%%%%%%%%%%%%%%%%%%%%%%%%%%%%%%%%%%%%%

\section{Introduction}

%The introduction should briefly place the study in a broad context and highlight why it is important. It should define the purpose of the work and its significance. The current state of the research field should be reviewed carefully and key publications should be cited. Please highlight controversial and diverging hypotheses when necessary. Finally, briefly mention the main aim of the work and highlight the main conclusions. As far as possible, please keep the introduction comprehensible to scientists outside your particular field of research. Citing a journal paper ~\cite{ref-journal}. And now citing a book reference ~\cite{ref-book}.

Shannon's mutual information \cite{Shannon1948} provides a well established, widely applicable tool to characterize the statistical relationship between two stochastic variables. Larger values of mutual information correspond to a stronger relationship between the instantiations of the two variables in each single trial. Whenever we study a system with more than two variables, the mutual information between any two subsets of the variables still quantifies the statistical dependencies between these two subsets; however, many scientific questions in the analysis of complex systems require a finer characterization of how all variables simultaneously interact~\cite{Ay2006,Bertschinger2013,Tononi1999,Tikhonov2015,Timme2014}. For example, two of the variables, $A$ and $B$, may carry either redundant or synergistic information about a third variable $C$~\cite{Pola2003,Schneidman2003,Latham2005}, but considering the value of the mutual information $I((A,B):C)$ alone is not enough to distinguish these qualitatively different information-carrying modes. To achieve this finer level of understanding, recent theoretical efforts have focused on decomposing the mutual information between two subsets of variables  into more specific information components (see e.g.~\cite{Timme2014,Williams2010,Bertschinger2014,Chicharro2017}). Nonetheless, a complete framework for the information-theoretic analysis of multivariate systems is still lacking.

Here we consider the analysis of trivariate systems, which is complex enough to present the fundamental challenges of going beyond bivariate analyses, and yet simple enough to provide hints about how these challenges might be addressed. 
Characterizing the fine structure of the interactions among three stochastic variables can improve the understanding of many interesting problems across different disciplines~\cite{Anastassiou2007,Ludtke2008,Watkinson2009,Faes2016}. For example, in the study of neural information processing many important questions can be cast as trivariate information analyses. Determining quantitatively how two neurons encode information about an external sensory stimulus~\cite{Pola2003,Schneidman2003,Latham2005} requires describing the dependencies between the stimulus and the activity of the two neurons. Determining how the stimulus information carried by a neural response relates to the animal's behaviour~\cite{Pitkow2015,Haefner2013,Panzeri2017} requires the analysis of the simultaneous three-wise dependencies among the stimulus, the neural activity and the subject's behavioral report. More generally, a thorough understanding of even the simplest information-processing systems would require the quantitative description of all different ways two inputs carry information about one output~\cite{Wibral2017}.

In systems where legitimate assumptions can be made about which variables act as sources of information and which variable acts as the target of information transmission,  the partial information decomposition (PID)~\cite{Williams2010} provides an elegant framework to decompose the mutual information that one or two (source) variables carry about the third (target) variable into a finer lattice of redundant, unique and synergistic information atoms. However, in many systems the \emph{a priori} classification of variables into sources and target is arbitrary, and limits the description of the distribution of information within the system~\cite{James2016}. Furthermore, even when one classification is adopted, the PID atoms do not characterize completely all the possible modes of information sharing between the sources and the target. For example, two sources can carry redundant information about the target irrespective of the strength of the correlations between them and, as a consequence,  the PID redundancy atom can be larger than zero even if the sources have no mutual information~~\cite{Harder2013,Bertschinger2013,Barrett2015}. Hence, the value of the PID redundancy measure cannot distinguish how the correlations between two variables contribute to the information that they share about a third variable.

In this paper, we address these limitations by extending the PID framework without introducing further axioms or assumptions about the trivariate structure to analyze.
We compare the atoms from the three possible PID lattices that are induced by the three possible choices for the target variable in the system. By tracking how the PID information modes change across different lattices, we move beyond the partial perspective intrinsic to a single PID lattice and unveil the finer structure common to all PID lattices. We find that all lattices can be fully described in terms of a unique minimal set of seven information-theoretic quantities, that is invariant to different classifications of the variables.

The first result of this approach is the identification of two nonnegative subatomic components of the redundant information that any pair of variables carries about the third variable. The first component, that we name source redundancy ($SR$), quantifies the part of the redundancy which arises from the correlations of the sources. The second component, that we name non-source redundancy ($NSR$), quantifies the part of the redundancy which is not related to the source correlations. Interestingly, we find that whenever the non-source redundancy is larger than zero then also the synergy is larger than zero. The second result is that the minimal set induces a unique nonnegative decomposition of the full joint entropy $H(X,Y,Z)$ of the system. This allows us to dissect completely the distribution of information of any trivariate system in a general way that is invariant with respect to the source/target classification of the variables.
To illustrate the additional insights of this new approach,  we finally apply our framework to paradigmatic examples, including discrete and continuous probability distributions. These applications confirm our intuitions and clarify the practical usefulness of the finer PID structure.

\section{Preliminaries and state of the art}
\label{sec:prelim}
Williams and Beer proposed an influential axiomatic  decomposition of the mutual information $I(X: (Y,Z))$ that two stochastic variables $Y, Z$ (the sources) carry about a third variable $X$ (the target) into the sum of four nonnegative atoms ~\cite{Williams2010}:
\begin{itemize}
\item $SI(X: \{Y;Z\}) $, which is the information about the target that is shared between the two sources (the redundancy);
\item $UI(X: \{Y \backslash Z\})$ and $UI(X: \{Z \backslash Y\})$, which are the separate pieces of information about the target that can be extracted from one of the sources, but not from the other;
\item $CI(X: \{Y;Z\}) $, which is the complementary information about the target that is only available when both of the sources are jointly observed (the synergy).
\end{itemize}

This construction is commonly known as the Partial Information Decomposition (PID). Sums of subsets of the four PID atoms provide the classical mutual information quantities between each of the sources and the target, $I(X:Y)$ and $I(X:Z)$, and the conditional mutual information quantities whereby one of the sources is the conditioned variable, $I(X:Z|Y)$ and $I(X:Y|Z)$. Such relationships are displayed with a color code in Fig.~\ref{fig:PID}.

\begin{figure}[t!]
\centering
\begin{subfigure}{.35\textwidth}
\includegraphics[width=\textwidth]{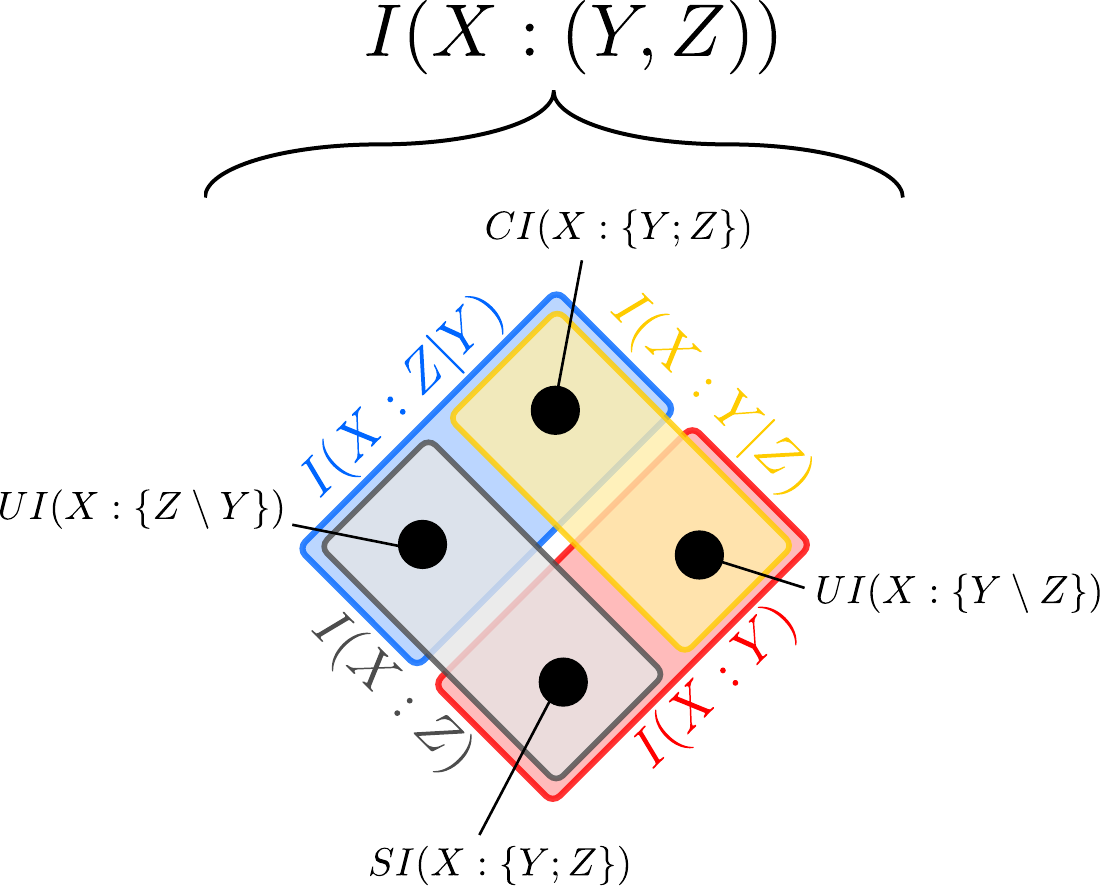}
\caption{}
\label{fig:PIDa}
\end{subfigure}\hspace{2ex}
\begin{subfigure}{.6\textwidth}
\begin{overpic}[trim={3cm 2cm 3.5cm 3cm},clip,width=\textwidth]{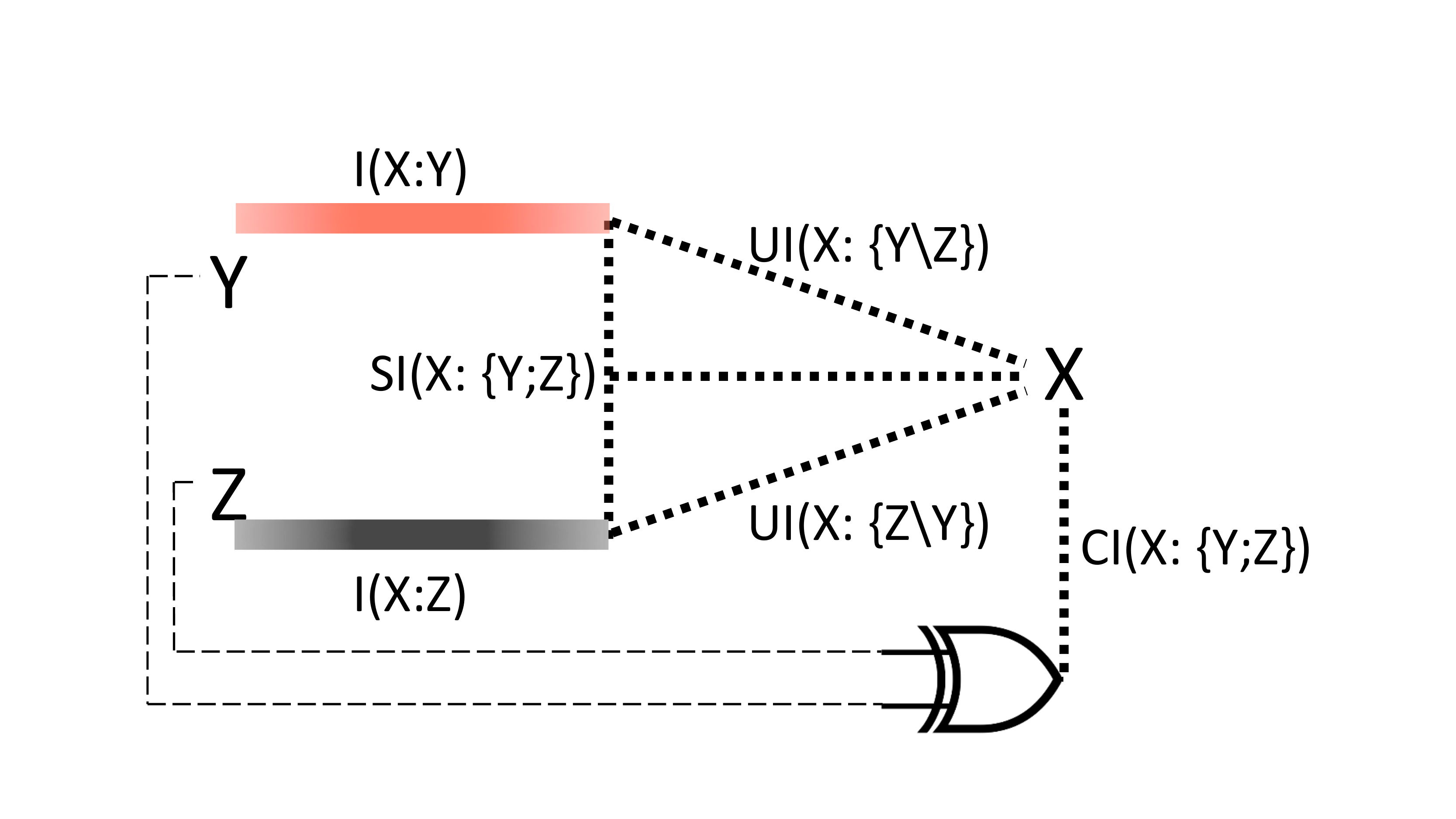}
\end{overpic}
\caption{}
\label{fig:PIDb}
\end{subfigure}
\caption{The Partial Information Decomposition as defined by Williams and Beer's axioms ~\cite{Williams2010}. \textbf{a)} The mutual information of the sources $Y, Z$ about the target $X$ is decomposed into four atoms: the redundancy $SI(X:\{Y;Z\})$, the unique informations $UI(X:\{Y \backslash Z\})$, $UI(X:\{Z \backslash Y\})$, and the synergy $CI(X:\{Y;Z\})$. The colored rectangles represent the linear equations that relate the four PID atoms to four Shannon information quantities. \textbf{b)} An exploded view of the allotment of information between the sources $Y, Z$ and the target $X$: each PID atom of panel \textbf{a)} corresponds to a thick dotted line, while the colored stripes represent the two pairwise mutual informations between each of the sources and the target (with the same color code as in \textbf{a)}). Each of the mutual informations splits into the sum of the redundancy with its corresponding unique information. The circuit-diagram symbol for the XOR operation is associated to the synergistic component $CI(X\{Y;Z\})$ only for illustration, as XOR is often taken as a paradigmatic example of synergistic interaction between variables.}
\label{fig:PID}
\end{figure}

The PID decomposition of Ref.~\cite{Williams2010} is based upon a number of axioms that do not determine univocally the value of the four PID atoms. The specific measure proposed in Ref.~\cite{Williams2010} to calculate such values has been questioned as it can lead to unintuitive results~\cite{Harder2013}, and thus many attempts have been devoted to finding alternative measures~\cite{Griffith2014, Harder2013, Griffith2014bis, Banerjee2015, Ince2017, Rauh2017} compatibly with an extended number of axioms, such as the \emph{identity axiom} proposed in~\cite{Harder2013}. Other work has studied in more detail the lattice structure that underpins the PID, indicating the duality between information gain and information loss lattices~\cite{Chicharro2017}. Even though there is no consensus on how to build partial information decompositions in systems with more than two sources, for trivariate systems the measures of redundancy, synergy and unique information defined in Ref.~\cite{Bertschinger2014} have not yet been questioned and have found wide acceptance (in this paper, we will in fact make use of these measures when a concrete implementation of the PID will be required).

However, even in the trivariate case there are open problems regarding the understanding of the PID atoms in relation to the interdependencies within the system. First, Harder et al.~\cite{Harder2013} pointed out that the redundant information shared between the sources about the target can intuitively arise from the following two qualitatively different modes of three-wise interdependence:

\begin{itemize}
\item the \emph{source redundancy}, which is redundancy which 'must already manifest itself in the mutual information between the sources'\footnote{Note that Ref.~\cite{Harder2013} interchangeably refers to the sources as 'inputs': we will discuss this further in Section ~\ref{sec:sourcered} when addressing the characterization of source redundancy.};
\item the \emph{mechanistic redundancy}, which can be larger than zero even if there is no mutual information between the sources.
\end{itemize}

As pointed out by Harder and colleagues~\cite{Harder2013}, a more precise conceptual and formal separation of these two kinds of redundancy still needs to be achieved, and presents fundamental challenges. The very notion that two statistically independent sources can nonetheless share information about a target was not captured by some earlier definitions of redundancy~\cite{Griffith2014,Griffith2014bis}. Further, several studies ~\cite{Barrett2015,Stramaglia2016} described the property that the PID measures of redundancy can be positive even when there are no correlations between the sources as undesired. On a different note, other authors~\cite{Wibral2017} pointed out that the two different notions of redundancy can define qualitatively different modes of information processing in (neural) input-output networks. 

Other issues were recently pointed out by James and Crutchfield~\cite{James2016}, who indicated that the very definition of the PID lattice prevents its use as a general tool for assessing the full structure of trivariate (let alone multi-variate) statistical dependencies. In particular, Ref.~\cite{James2016} considered dyadic and triadic systems, which underlie quite interesting and common modes of multivariate interdependencies. They showed that, even though the PID atoms are among the very few measures that can distinguish between the two kinds of systems, a PID lattice cannot allot the full joint entropy $H(X,Y,Z)$ of either system. The decomposition of the joint entropy in terms of information components that reflect qualitatively different interactions within the system has also been subject of recent research~\cite{Rosas2016,Ince2017}. 

In summary, the PID framework, in its current form, does not yet provide a satisfactorily fine and complete description of the distribution of information in trivariate systems. The PID atoms do assess trivariate dependencies better than Shannon's measures, but they cannot quantify interesting finer interdependencies within the system, such as the source redundancy that the sources share about the target. In addition, they are limited to describing the dependencies between the chosen sources and target, thus enforcing a certain perspective on the system that does not naturally suit all systems.

\section{More PID diagrams unveil finer structure in the PID framework}\label{sec:morepids}

To address the open problems described above, we begin by pointing out the feature of the PID lattice that underlies all the issues in the characterization of trivariate systems outlined in Sec.~\ref{sec:prelim}. As we illustrate in Fig.~\ref{fig:PIDa}, while a single PID diagram involves the mutual information quantities that one or both of the sources (in the figure, $Y$ and $Z$) carry about the target $X$, it does not contain the mutual information between the sources $I(Y:Z)$ and their conditional mutual information $I(Y:Z|X)$. This precludes the characterization of source redundancy with a single PID diagram, as it prevents any comparison between the elements of the PID and $I(Y:Z)$. Moreover, it also signals that a single PID lattice cannot account for the total entropy $H(X,Y,Z)$ of the system.

These considerations suggest that the inability of the PID framework to provide a complete information-theoretic description of trivariate systems is not a consequence of the axiomatic construction underlying the PID lattice. Instead, it follows from restricting the analysis to the limited perspective on the system that is enforced by classifying the  variables into sources and target when defining a PID lattice. We thus elaborate that significant progress can be achieved, \emph{without the addition of further axioms or assumptions to the PID framework}, if one just considers, alongside the PID diagram in Fig.~\ref{fig:PIDa}, the other two PID diagrams that are induced respectively by labeling $Y$ or $Z$ as the target in the system. When considering the resulting three PID diagrams (Fig.~\ref{fig:PIDs}),  the previously missing mutual information $I(Y:Z)$ between the original sources of the left-most diagram is now decomposed into PID atoms of the middle and the right-most diagrams in Fig.~\ref{fig:PIDs}, and the same happens with $I(Y:Z|X)$.

\begin{figure}[t!]
\centering
\begin{overpic}[trim={0 0cm 0 0cm},clip,width=\textwidth]{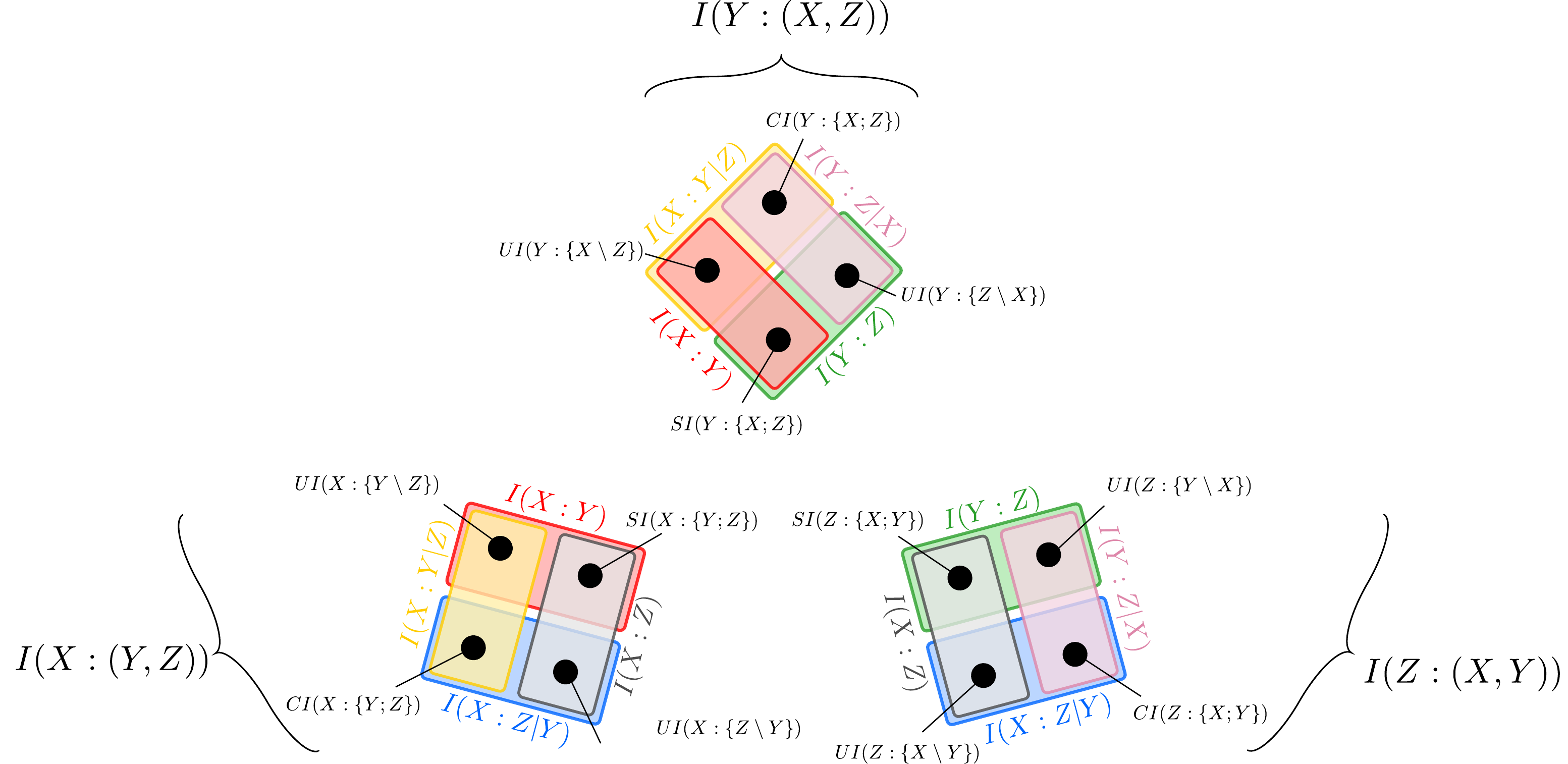}
\end{overpic}
\caption{The three possible PIDs of a trivariate probability distribution $p(x,y,z)$ that follow from the three possible choices for the target variable: on the left the target is $X$, in the middle it is $Y$ and on the right it is $Z$. The coloured rectangles highlight the linear relationships between the twelve PID atoms and the six Shannon information quantities. Note that the orientations of the PIDs for $I(X:(Y,Z))$ and $I(Z:(X,Y))$ are rotated with respect to the PID for $I(Y:(X,Z))$ to highlight their reciprocal relations, as will become more apparent in Fig.~\ref{fig:entropyconst}.}
\label{fig:PIDs}
\end{figure}

In the following we take advantage of this shift in perspective to resolve the finer structure of the PID diagrams and, at the same time, to generalize its descriptive power going beyond the current limited framework, where only the information that two (source) variables carry about the third (target) variable is decomposed. More specifically, even though the PID relies
on setting a partial point of view about the system, we will show that describing how the PID atoms change when we systematically rotate the choice of the PID target variable effectively overcomes the  limitations intrinsic to  one PID alone.

\subsection{The relationship between PID diagrams with different target selections}\label{sec:split}

\begin{figure}[t!]
\centering
\begin{overpic}[trim={0cm 0cm 0cm 0cm},clip,width=\textwidth]{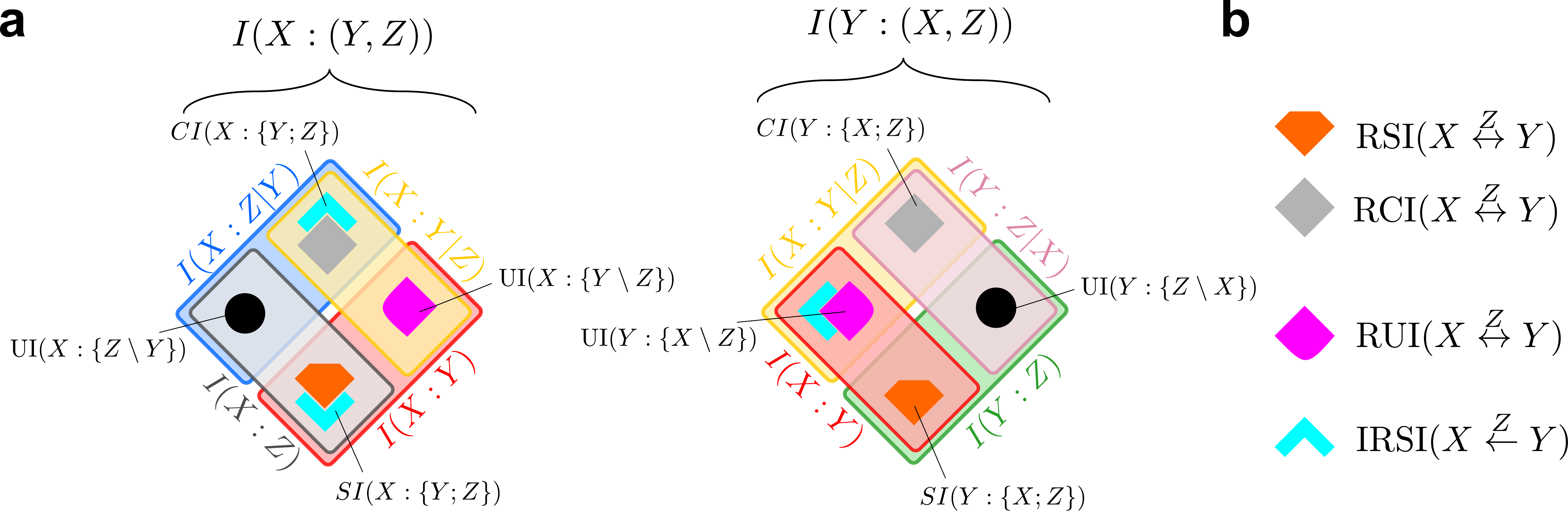}
\end{overpic}
\caption{\textbf{(a)} The relationships between the PID atoms from two diagrams with different target selections. When we swap a target and a source, the differences in the amount of redundancy, synergy, and unique information with respect to the third variable are not independent due to equations of the type of Eqs.~\ref{eq:pid1} and ~\ref{eq:pid2}. Here, we consider the PID diagram of $I(X: (Y,Z))$ (left) and $I(Y: (X,Z))$ (right) from Fig.~\ref{fig:PIDs}, under the assumption that $SI(Y: \{X;Z\})\leq SI(X: \{Y;Z\})$. \textbf{(b)} The reversible pieces of information $RSI(X \overset{\text{Z}}{\leftrightarrow} Y)$ (orange block), $RCI(X \overset{\text{Z}}{\leftrightarrow} Y)$ (gray block) and $RUI(X \overset{\text{Z}}{\leftrightarrow} Y)$ (magenta block) contribute to the same kind of atom across the two PID diagrams. The irreversible piece of information $IRSI(X \overset{\text{Z}}{\leftarrow} Y)$ (light blue block) contributes to different kinds of atom across the two PID diagrams. The remaining two unique information atoms (black dots) are not constrained by equations of the type of Eqs.~\ref{eq:pid1} and ~\ref{eq:pid2}. }
\label{fig:changepid}
\end{figure}

To identify the finer structure underlying all the PID diagrams in Fig.~\ref{fig:PIDs}, we first focus on the relationships between the PID atoms of two different diagrams, with the goal of understanding how to move from one perspective on the system to another. The key observation here is that, for each pair of variables in the system, their mutual information and their conditional mutual information given the third variable appear in two of the PID diagrams. This imposes some constraints relating the PID atoms in two different diagrams. For example, if we consider $I(X:Y)$ and $I(X:Y|Z)$, we find that:
\begin{gather}
I(X:Y)=SI(X:\{Y;Z\})+UI(X:\{Y\backslash Z\})=SI(Y:\{X;Z\})+UI(Y:\{X\backslash Z\}), \label{eq:pid1}\\
I(X:Y|Z)=CI(X:\{Y;Z\})+UI(X:\{Y\backslash Z\})=CI(Y:\{X;Z\})+UI(Y:\{X\backslash Z\}) \label{eq:pid2},
\end{gather}
where the first and second equality in each equation result from the decomposition of $I(X:(Y,Z))$ (left-most diagram in Fig.~\ref{fig:PIDs}) and $I(Y:(X,Z))$ (middle diagram in Fig.~\ref{fig:PIDs}), respectively. From Eq.~\ref{eq:pid1} we see that, when the roles of a target and a source are reversed (here, the roles of $X$ and $Y$), the difference in redundancy is the opposite of the difference in unique information with respect to the other source (here, $Z$). Similarly, Eq.~\ref{eq:pid2} shows that the difference in synergy is the opposite of the difference in unique information with respect to the other source. Combining these two equalities, we also see that the difference in redundancy is equal to the difference in synergy. Therefore, the equalities impose relationships across some PID atoms appearing in two different diagrams. 

These relationships are depicted in Figure ~\ref{fig:changepid}. The eight PID atoms appearing in the two diagrams can be expressed in terms of only six subatoms, due to the constraints of the form of Eqs.~\ref{eq:pid1} and~\ref{eq:pid2}. In particular, to select the smallest nonnegative pieces of information resulting from the constraints, we define:
\begin{subequations}
\begin{align}
RSI(X \overset{\text{Z}}{\leftrightarrow} Y)\coloneqq & \min[SI(X: \{Y;Z\}), SI(Y: \{X;Z\})],\label{eq:rsi1}\\
RCI(X \overset{\text{Z}}{\leftrightarrow} Y)\coloneqq & \min[CI(X: \{Y;Z\}), CI(Y: \{X;Z\})],\label{eq:rci1}\\
RUI(X \overset{\text{Z}}{\leftrightarrow} Y)\coloneqq & \min[UI(X:\{Y\backslash Z\}), UI(Y:\{X\backslash Z\})].\label{eq:rui1}
\end{align}
\end{subequations}
The above terms are called the Reversible Shared Information of $X$ and $Y$ considering $Z$ ($RSI(X \overset{\text{Z}}{\leftrightarrow} Y)$; the orange block in Fig.~\ref{fig:changepid}), the Reversible Complementary Information of $X$ and $Y$ considering $Z$ ($RCI(X \overset{\text{Z}}{\leftrightarrow} Y)$; the gray block in Fig.~\ref{fig:changepid}), and the Reversible Unique Information of $X$ and $Y$ considering $Z$ ($RUI(X \overset{\text{Z}}{\leftrightarrow} Y)$; the magenta block in Fig.~\ref{fig:changepid}). 
The attribute \emph{reversible} highlights that, when we reverse the roles of target and source between the two variables at the endpoints of the arrow in $RSI$, $RCI$, or $RUI$ (here, $X$ and $Y$), the reversible pieces of information are still included in the same type of PID atom (redundancy, synergy, or unique information with respect to the third variable). For example, the orange block in Fig.~\ref{fig:changepid} indicates a common amount of redundancy in both PID diagrams: as such, $RSI(X \overset{\text{Z}}{\leftrightarrow} Y)$ contributes both to redundant information that $Y$ and $Z$ share about $X$, and to redundant information that $X$ and $Z$ share about $Y$. By construction, these reversible components are symmetric in the reversed variables. Note that, when we reverse the role of two variables, the third variable (here, $Z$) remains a source and is thus put in the middle of our notation in Eqs.~\ref{eq:rsi1}, ~\ref{eq:rci1} and ~\ref{eq:rui1}.
We also define the Irreversible Shared Information $IRSI(X \overset{\text{Z}}{\leftarrow} Y)$ between $X$ and $Y$ considering $Z$ (the light blue block in Fig.~\ref{fig:changepid}) as follows:
\begin{equation}\label{eq:irsi1}
IRSI(X \overset{\text{Z}}{\leftarrow} Y)\coloneqq SI(X: \{Y;Z\})-RSI(X \overset{\text{Z}}{\leftrightarrow} Y).
\end{equation}
The attribute \emph{irreversible} in the above definition indicates that this piece of redundancy is specific to one of the two PIDs alone. More precisely, the uni-directional arrow in $IRSI(X \overset{\text{Z}}{\leftarrow} Y)$ indicates that this piece of information is a part of the redundancy with $X$ as a target, but it is not a part of the redundancy with $Y$ as a target\footnote{In this paper, directional arrows never represent any kind of \emph{causal} directionality: the PID framework is only capable to quantify statistical (correlational) dependencies.}. Correspondingly, at least one between $IRSI(X \overset{\text{Z}}{\leftarrow} Y)$ and $IRSI(Y \overset{\text{Z}}{\leftarrow} X)$ is always zero. More generally, $IRSI$ quantifies asymmetries between two different PIDs: for example, when moving from the left to the right PID in Fig.~\ref{fig:changepid}, the light blue block $IRSI(X \overset{\text{Z}}{\leftarrow} Y)$ indicates an equivalent amount of information that is lost for the redundancy $SI(X:\{Y;Z\})$ and the synergy $CI(X:\{Y;Z\})$ atoms, and is instead counted as a part of the unique information $UI(Y: \{X \backslash Z\})$ atom. In other words, assuming that the two redundancies are ranked as in Fig.~\ref{fig:changepid}, we find that:
\begin{subequations}
\label{eq:irsiexp}
\begin{align}
IRSI(X \overset{\text{Z}}{\leftarrow} Y) &= SI(X:\{Y;Z\})-SI(Y:\{X;Z\}) \label{eq:irsiexpa}\\ &= CI(X:\{Y;Z\})-CI(Y:\{X;Z\}) \label{eq:irsiexpb}\\ &= UI(Y:\{X\backslash Z\})-UI(X:\{Y \backslash Z\}).\label{eq:irsiexpc}
\end{align}
\end{subequations}
While the coarser Shannon information quantities that are decomposed in both diagrams in Fig.~\ref{fig:changepid}, namely $I(X : Y)$ and $I(X : Y|Z)$, are symmetric under swap of $X \leftrightarrow Y$, their PID decompositions (see Eqs.~\ref{eq:pid1} and ~\ref{eq:pid2}) are not: Eq.~\ref{eq:irsiexp} show that $IRSI$ quantifies the amount of this asymmetry.
More precisely, the PID decompositions of $I(X : Y)$ and $I(X : Y|Z)$ will preserve the $X \leftrightarrow Y$ symmetry if and only if $IRSI(X \overset{\text{Z}}{\leftarrow} Y)=IRSI(Y \overset{\text{Z}}{\leftarrow} X)=0$. Note that, in general, the differences of redundancies, of synergies, and of unique information terms are always constrained by equations such as Eq.~\ref{eq:irsiexp}. Hence, unlike for the reversible measures, we do not need to consider independent notions of irreversible synergy or irreversible unique information.

In summary, the four subatoms in Eqs.~\ref{eq:rsi1}, ~\ref{eq:rci1}, ~\ref{eq:rui1} and ~\ref{eq:irsi1}, together with the two remaining unique information terms (the black dots in Fig.~\ref{fig:changepid}), allow us to characterize both PIDs in Fig.~\ref{fig:changepid} and to understand how the PID atoms change when moving from one PID to another. Note that in all cases the blocks indicate amounts of information, but the interpretation of this information depends on the classification of variables as target and sources within each diagram. 

\subsection{Unveiling the finer structure of the PID framework}\label{sec:structure}

\begin{figure}[t!]
\centering
\includegraphics[width=\textwidth]{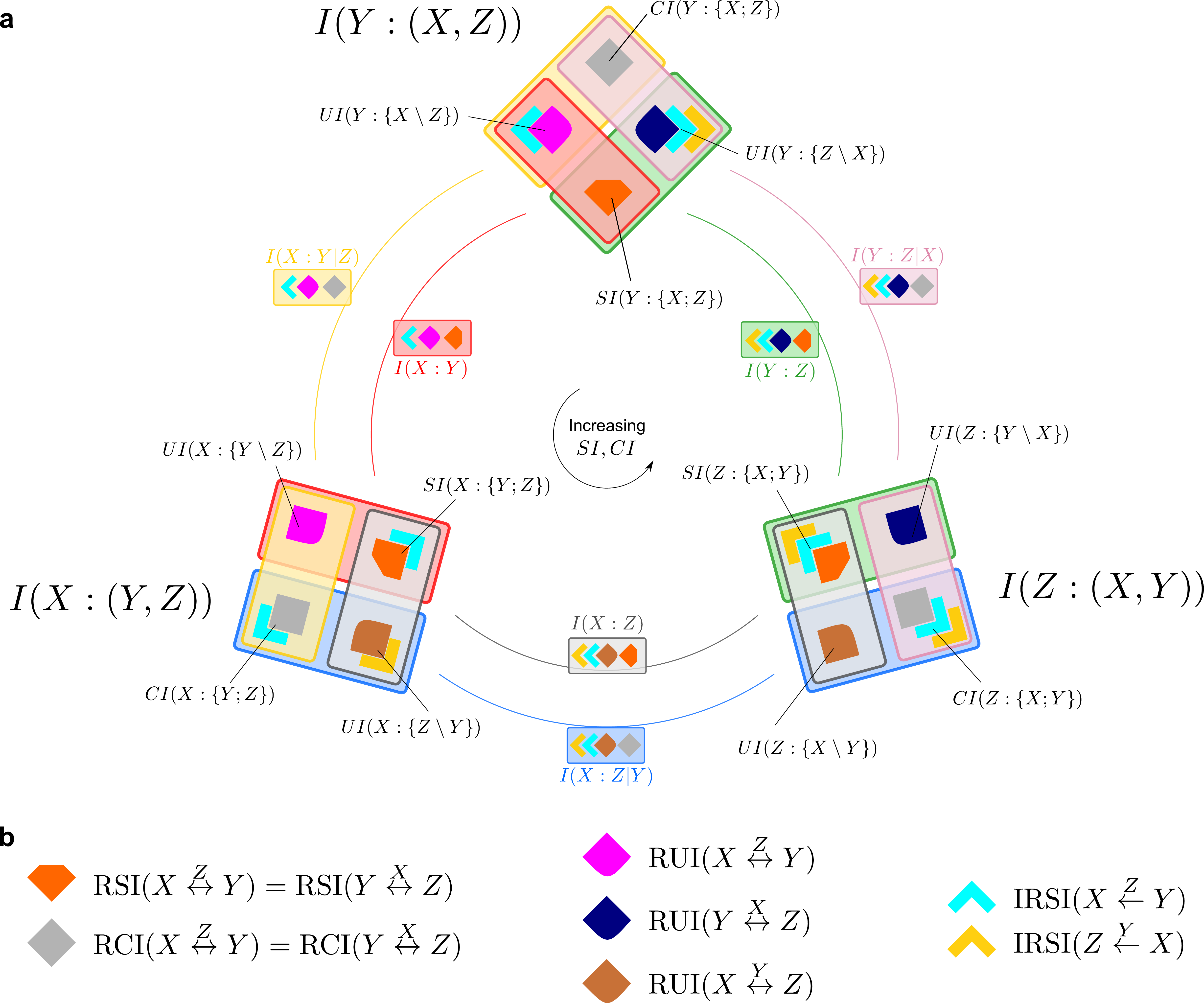}
\caption{Constructing the full structure of the three PID diagrams in Fig.~\ref{fig:PIDs} in terms of a minimal set of information subatoms. \textbf{a}: All the 12 PID atoms in Fig.~\ref{fig:PIDs} can be expressed as sums of seven independent PID subatoms that are displayed as coloured blocks in \textbf{b} (as in Fig.~\ref{fig:PIDs}, the orientations of the PIDs for $I(X:(Y,Z))$ and $I(Z:(X,Y))$ are rotated with respect to the PID for $I(Y:(X,Z))$ to highlight their reciprocal relations). Five of these subatoms are reversible pieces of information, that are included in the same kind of PID atom across different PID diagrams; the other two subatoms are the irreversible pieces of information, that can be included in different kinds of PID atom across different diagrams. Assuming, without loss of generality, that $SI(Y:\{X;Z\})\leq SI(X:\{Y;Z\})\leq SI(Z:\{X;Y\})$, the five reversible subatoms are: $RSI(X \overset{\text{Z}}{\leftrightarrow} Y)$ (orange), $RCI(X \overset{\text{Z}}{\leftrightarrow} Y)$ (gray), $RUI(X \overset{\text{Z}}{\leftrightarrow} Y)$ (magenta), $RUI(Y \overset{\text{X}}{\leftrightarrow} Z)$ (blue), and $RUI(X \overset{\text{Y}}{\leftrightarrow} Z)$ (brown). The two irreversible subatoms are $IRSI(X \overset{\text{Z}}{\leftarrow} Y)$ (light blue) and $IRSI(Z \overset{\text{Y}}{\leftarrow} X)$ (yellow).}
\label{fig:entropyconst}
\end{figure}

So far we have examined the relationships among the PID atoms corresponding to two different perspectives we hold about the system, whereby we reverse the roles of target and source between two variables in the system. We have seen that the PID atoms of different diagrams are not independent, as they are constrained by equations of the type of Eqs.~\ref{eq:pid1} and ~\ref{eq:pid2}. More specifically, the eight PID atoms of two diagrams can be expressed in terms of only six independent quantities, including reversible and irreversible pieces of information. The next question is how many subatoms we need to describe all the three possible PIDs (see Fig.~\ref{fig:PIDs}). Since there are six constraints, three equations of the type of Eq.~\ref{eq:pid1} and three equations of the type of Eq.~\ref{eq:pid2}, one may be tempted to think that the twelve PID atoms of all three PID diagrams can be expressed in terms of only six independent quantities. However, the six constraints are not independent: 
this is most easily seen from the symmetry of the co-information measure~\cite{McGill1954}, which is defined as the mutual information of two variables minus their conditional information given the third (e.g., Eq.~\ref{eq:pid1} minus Eq.~\ref{eq:pid2}). The co-information is invariant to any permutation of the variables, and this property highlights that only five of the six constraints are linearly independent. Accordingly, we will now detail how seven subatoms are sufficient to describe the whole set of PIDs: we call these subatoms the \emph{minimal subatoms' set} of the PID diagrams.

In Fig.~\ref{fig:entropyconst} we see how the minimal subatoms' set builds the PID diagrams. We assume, without loss of generality, that $SI(Y: \{X;Z\})\leq SI(X: \{Y;Z\})\leq SI(Z: \{X;Y\})$. Then, we consider the three possible instances of Eq.~\ref{eq:irsiexp} for the three possible choices of the target variable, and we find that the same ordering also holds for the synergy atoms: $CI(Y: \{X;Z\})\leq CI(X: \{Y;Z\})\leq CI(Z: \{X;Y\})$. This property is related to the invariance of the co-information: indeed, Ref.~\cite{Williams2010} indicated that the co-information can be expressed as the difference between the redundancy and the synergy within each PID diagram, i.e.
\begin{gather}
\label{eq:coi}
coI(X;Y;Z)= SI(i: \{j;k\})-CI(i: \{j;k\}),
\end{gather}
for any assignment of $X$, $Y$, $Z$, to $i$, $j$, $k$. 

These ordering relations are enough to understand the nature of the minimal subatoms' set: we start with the construction of the three redundancies, which can all be expressed in terms of the smallest $RSI$ and two subsequent increments. In Fig.~\ref{fig:entropyconst}, these correspond respectively to $RSI(X \overset{\text{Z}}{\leftrightarrow} Y)= SI(Y: \{X;Z\})$ (orange block), $IRSI(X \overset{\text{Z}}{\leftarrow} Y)$ (light blue block) and $IRSI(Z \overset{\text{Y}}{\leftarrow} X)$ (yellow block). In parallel, we can construct the three synergies with the smallest $RCI$ and the same increments used for the redundancies. In Fig. ~\ref{fig:entropyconst}, these correspond respectively to $RCI(X \overset{\text{Z}}{\leftrightarrow} Y)= CI(Y: \{X;Z\})$ (gray block) and the same two $IRSI$ used before. To construct the unique information atoms, it is sufficient to further consider the three independent $RUI$ defined by taking all possible permutations of $X$, $Y$ and $Z$ in Eq.~\ref{eq:rui1}.
In Fig.~\ref{fig:entropyconst}, these correspond to $RUI(X \overset{\text{Z}}{\leftrightarrow} Y)=UI(X:\{Y\backslash Z\})$ (magenta block), $RUI(X \overset{\text{Y}}{\leftrightarrow} Z)=UI(Z:\{X\backslash Y\})$ (brown block), and $RUI(Y \overset{\text{X}}{\leftrightarrow} Z)=UI(Z:\{Y\backslash X\})$ (blue block). 
We thus see that, in total, seven minimal subatoms are enough to build the three PID diagrams of any system. Among these seven building blocks, five are reversible pieces of information, i.e. they contribute to the the same kind of PID atom across different PID diagrams; the other two are irreversible pieces of information, that contribute to different kinds of PID atom across different diagrams. The complete minimal set can only be determined when all three PIDs are jointly considered and compared: as shown in Fig.\ref{fig:changepid}, pairwise PIDs' comparisons can at most distinguish two subatoms in any redundancy (or synergy), while the three-wise PIDs' comparison discussed above allowed us to discern three subatoms in $SI(Z:\{X;Y\})$ (and $CI(Z:\{X;Y\})$; see also Fig.\ref{fig:entropyconst}).

Importantly, while the definition of a single PID lattice relies on the specific perspective adopted on the system, which labels two variables as the sources and one variable as the target, the decomposition in Fig.~\ref{fig:entropyconst} is invariant with respect to the classification of the variables. As described above, it only relies on computing all three PID diagrams and then using the ordering relations of the atoms, without any need to classify the variables \emph{a priori}. As illustrated in Fig.~\ref{fig:entropyconst}, the decomposition of the mutual information and conditional mutual information quantities in terms of the subatoms is also independent of the PID adopted. Our invariant minimal set in Fig.~\ref{fig:entropyconst} thus extends the descriptive power of the PID framework beyond the limitations that were intrinsic to considering an individual PID diagram.
In the next sections, we will show how the invariant minimal set can be used to identify the part of the redundant information about a target that specifically arises from the mutual information between the sources (the source redundancy), and to decompose the total entropy of any trivariate system.

Remarkably, the decomposition in Fig.~\ref{fig:entropyconst} does not rely on any extension of Williams and Beer's axioms. We unveiled finer structure underlying the PID lattices just by considering more PID lattices at a time and comparing PID atoms across different lattices. We further remark that the decomposition in Fig.~\ref{fig:entropyconst} does not rely in any respect on the specific definition of the PID measures that is used to calculate the PID atoms: it only relies on the axiomatic PID construction presented in Ref.~\cite{Williams2010}.

\section{Quantifying source redundancy}
\label{sec:sourcered}

The structure of the three PID diagrams that was unveiled with the construction in Fig.~\ref{fig:entropyconst} enables a finer characterization of the modes of information distribution among three variables than what has previously been possible. In particular, we will now address the open problem of quantifying the \emph{source redundancy}, i.e. the part of the redundancy that 'must already manifest itself in the mutual information between the sources'~\cite{Harder2013}. Consider for example the redundancy $SI(X:\{Y;Z\})$ in Fig.~\ref{fig:entropyconst}: it is composed by $RSI(X \overset{\text{Z}}{\leftrightarrow} Y)$ (orange block) and $IRSI(X \overset{\text{Z}}{\leftarrow} Y)$ (light blue block). We can check which of these subatoms are shared with the mutual information of the sources $I(Y:Z)$. To do this, we have to move from the middle PID diagram in Fig.~\ref{fig:entropyconst}, that contains $SI(X:\{Y;Z\})$, to any of the other two diagrams, that both contain $I(Y:Z)$. Consistently, in these other two diagrams $I(Y:Z)$ is composed by the same four subatoms (the orange, the light blue, the yellow and the blue block), and the only difference across diagrams is that these subatoms are differently distributed between unique information and redundancy PID atoms. In particular, we can see that both the orange and the light blue block which make up $SI(X:\{Y;Z\})$ are contained in $I(Y:Z)$. Thus, whenever any of them is nonzero, we know \emph{at the same time} that $Y$ and $Z$ share some information about $X$ (i.e., $SI(X:\{Y;Z\})>0$) and that there are correlations between $Y$ and $Z$ (i.e., $I(Y:Z)>0$). Accordingly, in the scenario of Fig.~\ref{fig:entropyconst} the entire redundancy $SI(X:\{Y;Z\})$ is explained by the mutual information of the sources: all the redundant information that $Y$ and $Z$ share about $X$ arises from the correlations between $Y$ and $Z$.

If we then consider the redundancy $SI(Y:\{X;Z\})$, that coincides with the orange block, we also find that it is totally explained in terms of the mutual information between the corresponding sources $I(X:Z)$, which indeed contains an orange block. However, if we consider the third redundancy $SI(Z:\{X;Y\})$, that is composed by an orange, a light blue and a yellow block, we find that only the orange and the light blue block contribute to $I(X:Y)$, while the yellow does not. This means that if $IRSI(Z \overset{\text{Y}}{\leftarrow} X)>0$ (yellow block), then $X$ and $Y$ can share information about $Z$ (i.e., $SI(Z:\{X;Y\})>0$) even if there is no mutual information between the sources $X$ and $Y$ (i.e., $I(X:Y)=0$). Following this reasoning, we define the source redundancy that two sources $S_1$ and $S_2$ share about a target $T$ as:
\begin{equation}\label{eq:SR2}
\begin{split}
SR(T:\{S_1;S_2\}) &\coloneqq \max\{RSI(T \overset{S_2}{\leftrightarrow} S_1) \ ,\ RSI(T \overset{S_1}{\leftrightarrow} S_2) \}\\
&=\max\{\min\left[SI(T:\{S_1;S_2\}),SI(S_1:\{S_2;T\})\right],\\
&\qquad\qquad\min\left[SI(T:\{S_1;S_2\}),SI(S_2:\{S_1;T\})\right] \}.
\end{split}
\end{equation}
One can easily verify that Eq.~\ref{eq:SR2} identifies the blocks that belong to both $SI(T,\{S_1,S_2\})$ and $I(S_1:S_2)$ in Fig.~~\ref{fig:entropyconst}, for any choice of sources and target (for instance, $T=Z$, $S_1=X$ and $S_2=Y$).
This definition can be justified as follows: each $RSI$ measure in Eq.~\ref{eq:SR2} compares $SI(T:\{S_1;S_2\})$ with one of the other two redundancies that are contained in the mutual information between the sources $I(S_1:S_2)$ (namely, $SI(S_1:\{S_2;T\})$ and $SI(S_2:\{S_1;T\})$). Some of the subatoms included in $I(S_1:S_2)$ are contained in one of these two redundancies, but not in the other, as they move to the unique information mode when we change PID.
Therefore, by taking the maximum in Eq.~\ref{eq:SR2} we ensure that $SR(T:\{S_1;S_2\})$ captures all the common subatoms of $I(S_1:S_2)$ and $SI(T:\{S_1;S_2\})$.
In a complementary way, we can define the non-source redundancy that two sources share about a target:
\begin{equation}\label{eq:NSR}
NSR(T:\{S_1;S_2\}) \coloneqq SI(T:\{S_1;S_2\})-SR(T:\{S_1;S_2\}).
\end{equation}

\begin{figure}[t!]
\centering
\begin{overpic}[trim={2cm 2cm 3.5cm 3cm},clip,width=.8\textwidth]{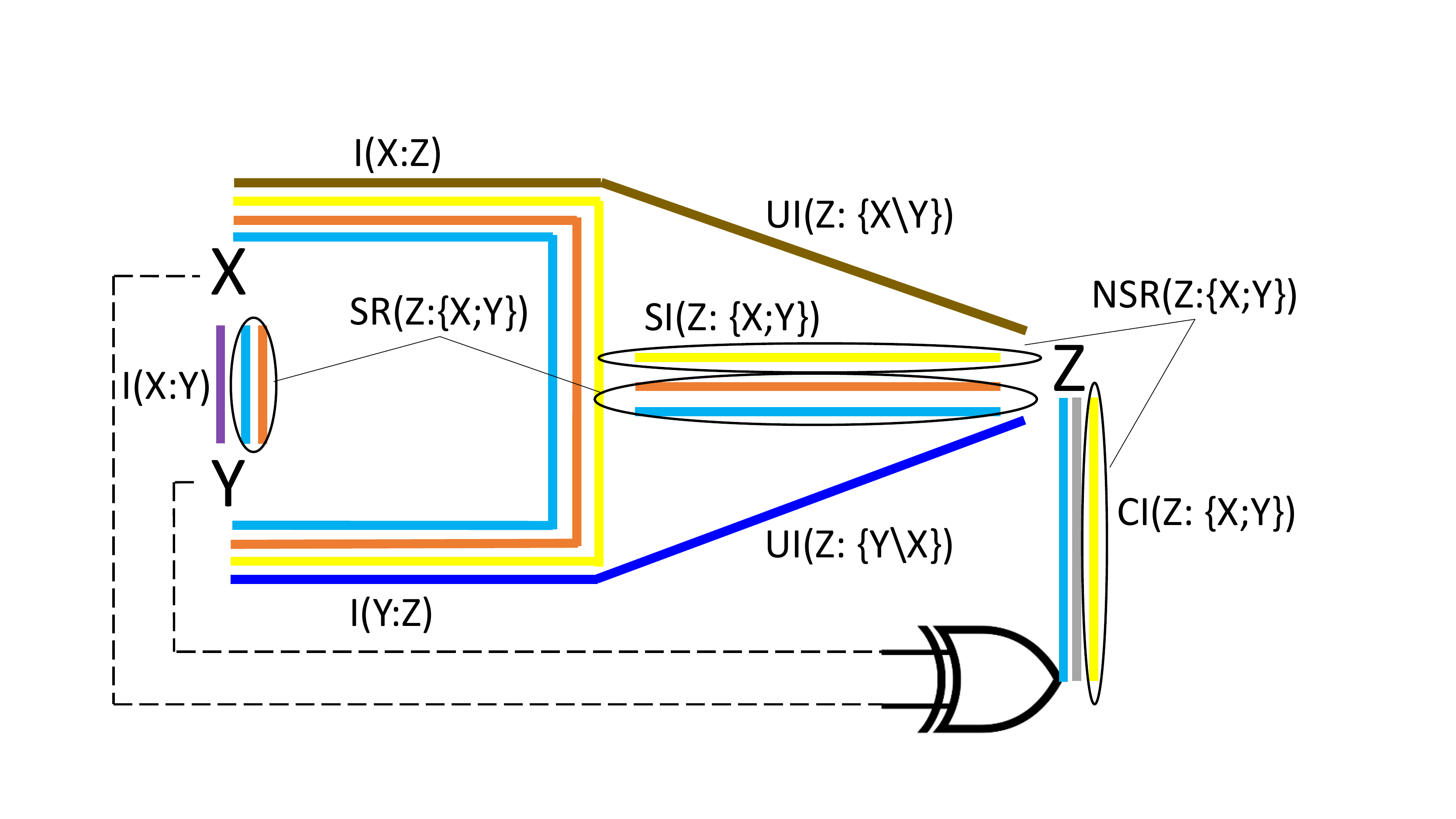}
\end{overpic}
\caption{Exploded view of the information that two variables $X, Y$ carry about a third variable $Z$. The mutual informations of each source about the target are decomposed into PID atoms as in Fig.~\ref{fig:PIDb}, but the PID atoms are now further decomposed in terms of the minimal subatoms' set: the thick colored lines represent the subatoms with the same colour code as in Fig.~\ref{fig:entropyconst}. Here, we assume that the variables are ordered as in Fig.~\ref{fig:entropyconst}. The finer structure of the PID atoms allows us to identify the source redundancy $SR(Z:\{X;Y\})$ (orange line + light blue line) as the part of the full redundancy $SI(Z:\{X;Y\})$ that is apparent in the mutual information between the sources $I(X:Y)$. Instead, the amount of information in the non-source redundancy $NSR(Z:\{X;Y\})$ (yellow line) also appears in the synergy $CI(Z:\{X;Y\})$.
}
\label{fig:circuit_sr}
\end{figure}

\subsection{The difference between source and non-source redundancy}

Eqs.~\ref{eq:SR2} and ~\ref{eq:NSR} show how we can split the redundant information that two sources share about a target into two nonnegative information components: when the source-redundancy $SR$ is larger than zero there are also correlations between the sources, while the non-source redundancy $NSR$ can be larger than zero even when the sources are independent. $SR$ is thus seen to quantify the \emph{pairwise correlations between the sources that also produce redundant information about the target}: this discussion is pictorially summarized in Fig.~\ref{fig:circuit_sr}. In particular, the source redundancy $SR$ is clearly upper-bounded by the mutual information between the sources, i.e.
\begin{equation}\label{eq:srbound}
SR(T:\{S_1;S_2\}) \leq I(S_1:S_2).
\end{equation}
On the other hand, $NSR$ does not arise from the pairwise correlations between the sources: let us calculate $NSR$ in a paradigmatic example that was proposed in Ref.~\cite{Harder2013} to remark the subtle possibility that two statistically independent variables can share information about a third variable. Suppose that $Y$ and $Z$ are uniform binary random variables, with $Y \independent Z$, and $X$ is deterministically fixed by the relationship $X= Y \wedge Z$. Here, $SI(X:\{Y;Z\}) \approx 0.311$ bit (according to different measures of redundancy ~\cite{Harder2013,Bertschinger2014}) even if $I(Y:Z)=0$. Indeed, from our definitions in Eqs. ~\ref{eq:SR2} and ~\ref{eq:NSR} we find that, since $SI(Y:\{X;Z\}), SI(Z:\{X;Y\})\leq I(Y:Z)=0$, here $NSR(X:\{Y;Z\})=SI(X:\{Y;Z\})>0$ even though $I(Y:Z)=0$. We will comment more extensively on this instructive example in Section ~\ref{sec:characterizing}.

Interestingly, the non-source redundancy $NSR$ is a part of the redundancy that is related to the synergy of the same PID diagram. Indeed, two of the three possible $NSR$ defined in Eq.~\ref{eq:NSR} are always zero, and the third
can be larger than zero if and only if the yellow block in Fig.~\ref{fig:entropyconst} is larger than zero. From Fig.~\ref{fig:entropyconst} and Fig.~\ref{fig:circuit_sr} we can thus see that, whenever we find positive non-source redundancy in a PID diagram, the same amount of information (the yellow block) is also present in the synergy of that diagram. 
Thus, while there is source redundancy if and only if there is mutual information between the sources, the existence of non-source redundancy is a sufficient (though not necessary) condition for the existence of synergy. We can thus interpret $NSR$ as \emph{redundant information about the target that implies that the sources carry synergistic information about the target}: we give a graphical characterization of $NSR$ in Fig.~\ref{fig:circuit_sr}.

In the specific examples considered in~\cite{Harder2013}, where the underlying causal structure of the system is such that the sources always generate the target, the non-source redundancy can indeed be associated with the notion of 'mechanistic redundancy' that was introduced in that work: the causal mechanisms connecting the target with the sources induce a non-zero $NSR$ that contributes to the redundancy independently of the correlations between the sources. In general, since the causal structure of the analyzed system is unknown, it is impossible to quantify 'mechanistic redundancy' with statistical measures, while it is always possible to quantify and interpret the non-source redundancy as described in Section~\ref{sec:sourcered}.

In section~\ref{sec:characterizing} we will examine concrete examples to show how our definitions of source and non-source redundancy refine the information-theoretic description of trivariate systems, as they quantify qualitatively different ways that two variables can share information about a third.

We conclude this Section with more general comments about our quantification of source redundancy. We note that the arguments used to define the source redundancy in Sec.~\ref{sec:sourcered} can be equally used to study common or exclusive information components of other PID terms. For example, we can identify the magenta subatom as the component of the mutual information between the sources $X$ and $Y$ that cannot be related to their redundant information about $Z$. Similarly, we could consider which part of a synergy is related to the conditional mutual information between the sources. 

\section{Decomposing the joint entropy of a trivariate system}

Understanding how information is distributed in trivariate systems should also provide a descriptive allotment of all parts of the joint entropy $H(X,Y,Z)$ ~\cite{James2016,Rosas2016}. For comparison, Shannon's mutual information enables a semantic decomposition of the bivariate entropy $H(X,Y)$ in terms of univariate conditional entropies and $I(X:Y)$, that quantifies shared fluctuations (or covariations) between the two variables~\cite{Rosas2016}:
\begin{equation}\label{eq:entropydec2}
H(X,Y)=I(X:Y)+H(X|Y)+H(Y|X).
\end{equation}
However, in spite of recent efforts~\cite{Rosas2016}, a univocal descriptive decomposition of the trivariate entropy $H(X,Y,Z)$ is still missing to date. Since the PID axioms in Ref.~\cite{Williams2010} decompose mutual information quantities, one might hope that the PID atoms could also provide a descriptive entropy decomposition.
Yet, at the beginning of Section~\ref{sec:morepids}, we pointed out that a single PID lattice does not include the mutual information between the sources and their conditional mutual information given the target: this suggests that a single PID lattice cannot in general contain the full $H(X,Y,Z)$. More concretely, Ref.~\cite{James2016} has recently suggested precise examples of trivariate dependencies where a single PID lattice cannot account for, and thus describe the parts of, the full $H(X,Y,Z)$.

These examples compare the dyadic and triadic dependencies described in Fig.~\ref{fig:dyadic}. 
Both kinds of dependencies underlie common modes of information sharing among three and more variables~\cite{James2016}.
Ref.~\cite{James2016} remarked that the atoms of a single PID diagram are indeed able to distinguish between dyadic and triadic dependencies, but that such atoms only sum up to two of the three bits of the full $H(X,Y,Z)$.

\begin{figure}[t!]
\centering
\begin{subfigure}{.49\linewidth}
\centering
\begin{overpic}[trim={4cm 3cm 4cm 1cm},clip,width=\linewidth]{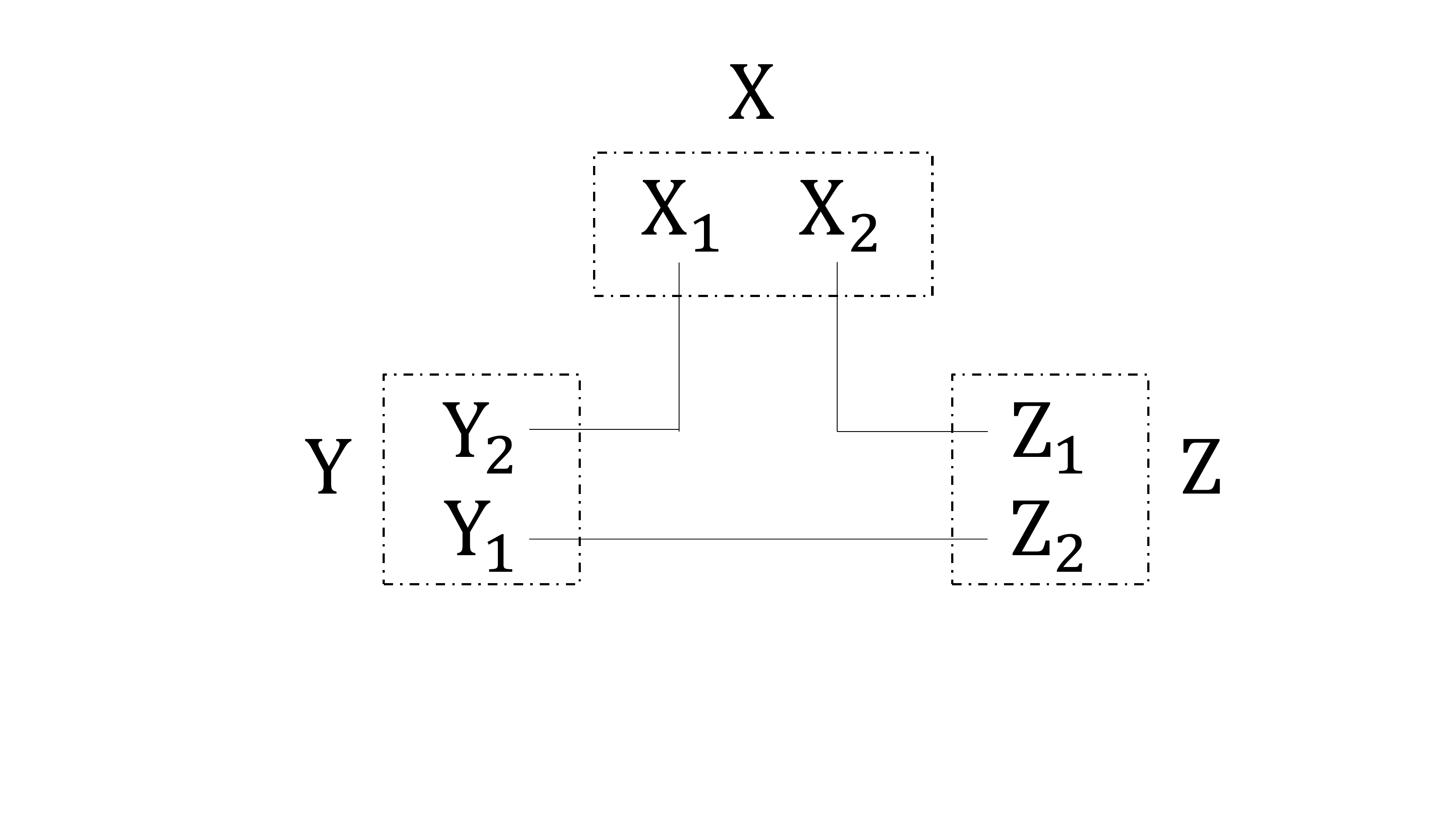}
\end{overpic}
\begin{tabular}{ c  c  c c }
  \hline			
x & y & z & $p(x,y,z)$ \\
\hline
0 & 0 & 0 & 1/8 \\
0 & 2 & 1 & 1/8 \\
1 & 0 & 2 & 1/8 \\
1 & 2 & 3 & 1/8 \\
2 & 1 & 0 & 1/8 \\
2 & 3 & 1 & 1/8 \\
3 & 1 & 2 & 1/8 \\
3 & 3 & 3 & 1/8 \\
  \hline
\end{tabular}
\caption{}
\label{tab:dyadic}
\end{subfigure}
\begin{subfigure}{.49\linewidth}
\centering
\begin{overpic}[trim={4cm 3cm 4cm 1cm},clip,width=\linewidth]{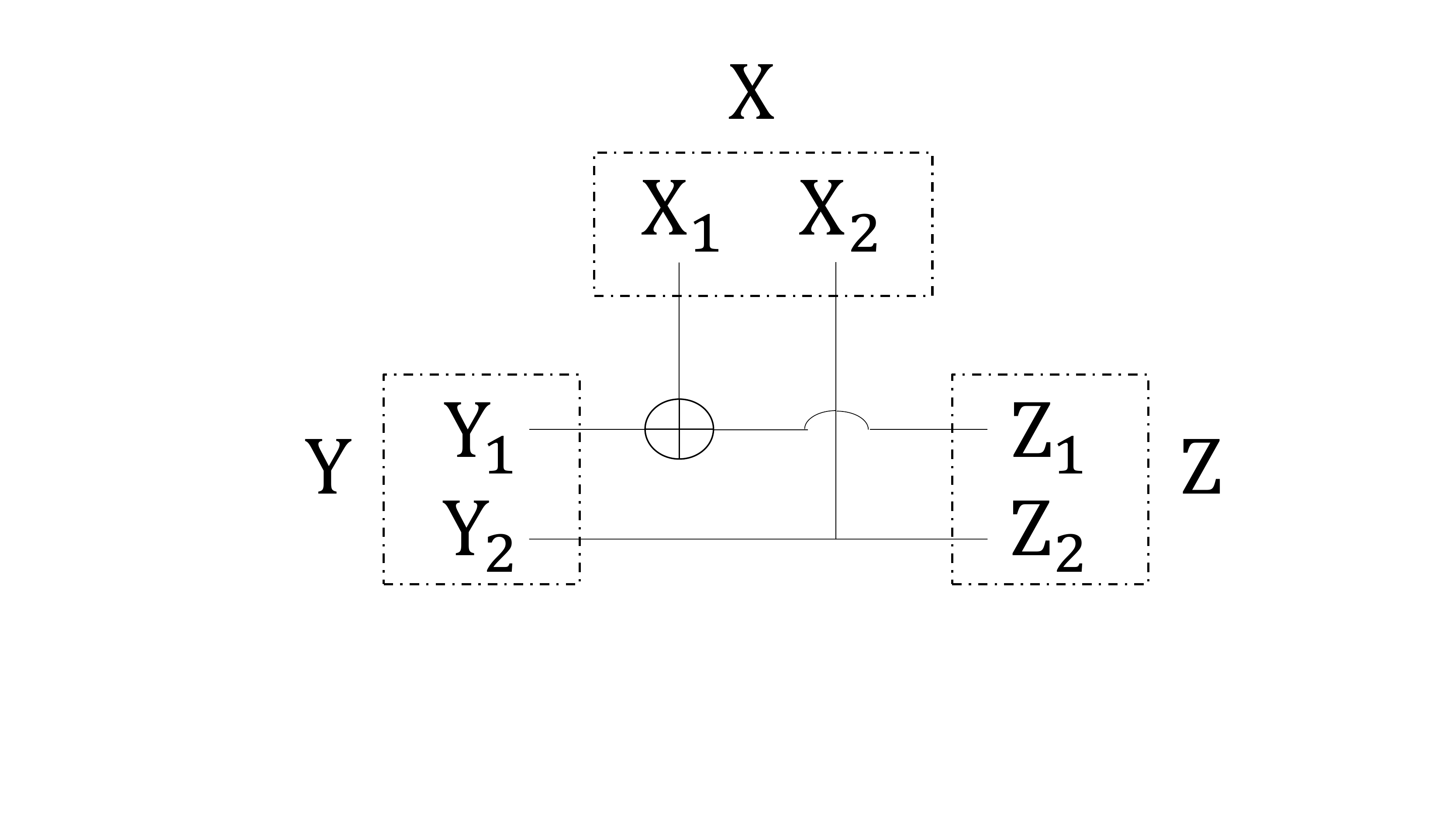}
\end{overpic}
\begin{tabular}{ c  c  c c }
  \hline			
x & y & z & $p(x,y,z)$ \\
\hline
0 & 0 & 0 & 1/8 \\
1 & 1 & 1 & 1/8 \\
0 & 2 & 2 & 1/8 \\
1 & 3 & 3 & 1/8 \\
2 & 0 & 2 & 1/8 \\
3 & 1 & 3 & 1/8 \\
2 & 2 & 0 & 1/8 \\
3 & 3 & 1 & 1/8 \\
  \hline
\end{tabular}
\caption{}
\label{tab:triadic}
\end{subfigure}
\caption{Dyadic and triadic statistical dependencies in a trivariate system, as defined in Ref.~\cite{James2016}. 
The tables display the non-zero probability values $p(x,y,z)$ as a function of the possible outcomes of the three stochastic variables $X, Y, Z$ with domain $\{0, 1 , 2, 3 \}$. $X\sim (X_1, X_2)$, $Y\sim(Y_1, Y_2)$, $Z\sim(Z_1, Z_2)$ (the symbol $\sim$ here means 'is distributed as'), where $X_1, X_2, Y_1, Y_2, Z_1, Z_2$ are  binary uniform random variables. \textbf{a)}: The underlying rules that give rise to dyadic dependencies are $X_1=Y_2$, $Y_1=Z_2$, $Z_1=X_2$. \textbf{b)}: The underlying rules that give rise to triadic dependencies are $X_1=Y_1 \oplus Z_1 $, $X_2=Y_2=Z_2$.
}
\label{fig:dyadic}
\end{figure}

We will now show how the missing allotment of the third bit of entropy in those systems is not due to intrinsic limitations of the PID axioms, but just to the limitations of considering a single PID diagram at a time --- the common practice in the literature so far. More generally, we will be able to allot and describe the full entropy $H(X,Y,Z)$ of any trivariate system in terms of the novel finer structure unveiled in Section~\ref{sec:morepids}.

\subsection{The finer structure of the entropy $H(X,Y,Z)$}
The minimal subatoms' set that we illustrated in Fig.~\ref{fig:entropyconst} allowed us to decompose all three PID lattices of a generic system. However, to fully describe the distribution of information in trivariate systems, we also wish to find a generalization of Eq.~\ref{eq:entropydec2} to the trivariate case, i.e. to decompose the full trivariate entropy $H(X,Y,Z)$ in terms of univariate conditional entropies and PID quantities.
With this goal in mind, we first subtract from $H(X,Y,Z)$ the terms which describe statistical fluctuations of only one variable (conditioned on the other two). The sum of these terms was indicated as $H_{(1)}$ in Ref.~\cite{Rosas2016}, and there quantified as
\begin{equation}
H_{(1)}=H(X|Y,Z)+H(Y|X,Z)+H(Z|Y,X). \label{eq:h1}
\end{equation}
This subtraction is useful because $H_{(1)}$ is a part of the total entropy which does not overlap with any of the 12 PID atoms in Fig.~\ref{fig:PIDs}. The remaining entropy $H(X,Y,Z)-H_{(1)}$ was defined as the dual total correlation in Ref.~\cite{Han1978} and recently considered in Ref.~\cite{Rosas2016}:
\begin{equation}\label{eq:entropydec}
DTC\equiv  H(X,Y,Z)-H_{(1)}.
\end{equation}
$DTC$ quantifies joint statistical fluctuations of more than one variable in the system. A simple calculation yields
\begin{equation}\label{eq:dtc}
DTC=I(X:Y|Z)+I(Y:Z|X)+I(X:Z|Y)+ coI(X;Y;Z),
\end{equation}
which is manifestly invariant under permutations of $X$, $Y$ and $Z$, and shows that $DTC$ can be written as a sum of some of the $12$ PID atoms. For example, expressing the co-information as the difference $I(X:Z) - I(X:Z|Y)$, we can arbitrarily use the four atoms from the left-most diagram in Fig.~\ref{fig:PIDs} to decompose the sum $I(X:Y|Z)+I(X:Z)$ and then add $UI(Y:\{Z\backslash X\})+CI(Y:\{X;Z\})$ from the middle diagram to decompose $I(Y:Z|X)$. If we then plug this expression of $DTC$ in Eq.~\ref{eq:entropydec}, we achieve a decomposition of the full entropy of the system in terms of $H_{(1)}$ and PID quantities: 
\begin{align}\label{eq:entropydecfine}
H(X,Y,Z) & = H_{(1)}+SI(X: \{Y;Z\} )+UI(X: \{Y \backslash Z\} ) + UI(X: \{Z \backslash Y\} )+ \nonumber \\ & + CI(X: \{Y;Z\} )+ UI(Y:\{Z\backslash X\})+CI(Y:\{X;Z\}),
\end{align}
which provides a nonnegative decomposition of the total entropy of any trivariate system. However, this decomposition is not unique, since the co-information can be expressed in terms of different pairs of conditional and unconditional mutual informations, according to Eq.~\ref{eq:coi}. This arbitrariness strongly limits the descriptive power of this kind of entropy decompositions, because the PID atoms on the RHS can only be interpreted within individual PID perspectives.

\begin{figure}[t!]
\centering
\includegraphics[width=.7\textwidth]{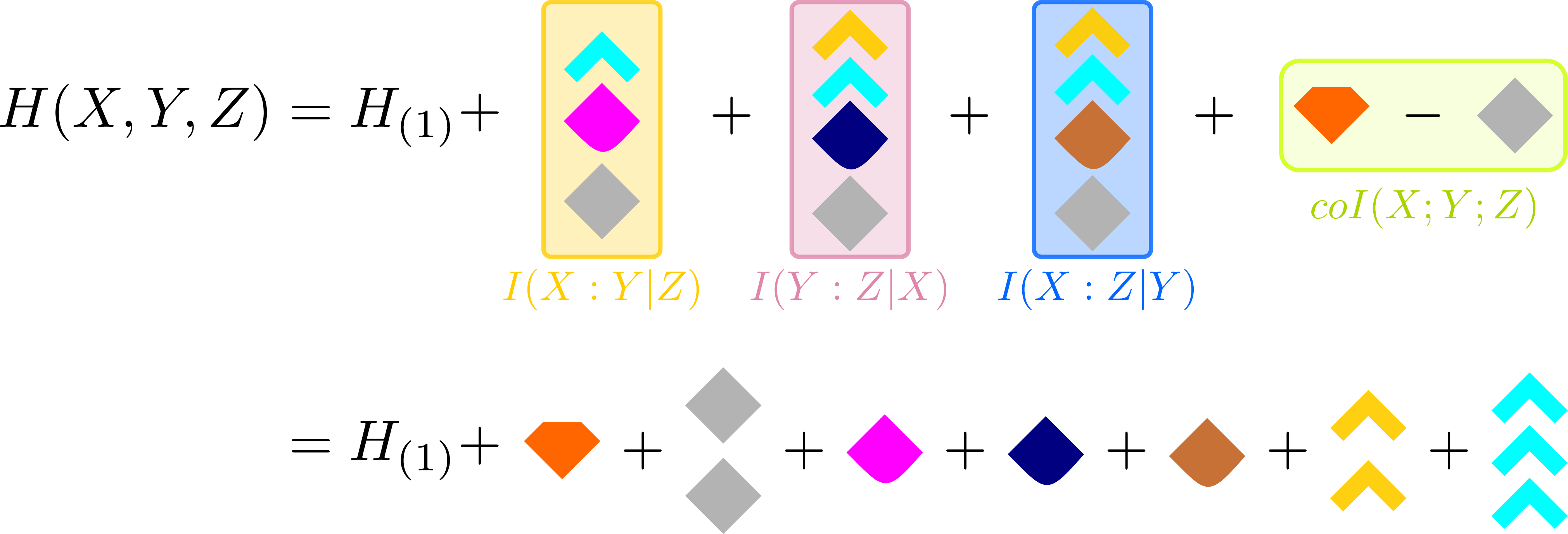}
\caption{The joint entropy of the system $H(X,Y,Z)$ is decomposed in terms of the minimal set identified in Fig.~\ref{fig:entropyconst}, once the univariate fluctuations quantified with $H_{(1)}$ have been subtracted out (see Eq.~\ref{eq:entropydecfinal}). As in Fig.~\ref{fig:entropyconst}, we assume without loss of generality $SI(Y: \{X;Z\})\leq SI(Z: \{X;Y\})\leq SI(X: \{Y;Z\})$. The finer PID structure unveiled in this work enables a general entropy decomposition in terms of quantities that can be interpreted without relying on a specific PID point of view, even though they have been defined within the PID framework. The colored areas represent the Shannon information quantities that are included in the $DTC$ of Eq.~\ref{eq:dtc}, with the same color code of Fig.~\ref{fig:PIDs}.}
\label{fig:dec_entropy_final}
\end{figure}

To address this issue, we construct a less arbitrary entropy decomposition by using the invariant minimal subatoms' set that was presented in Fig.~\ref{fig:entropyconst}: importantly, that set can be interpreted without specifying an individual, and thus partial, PID point of view that we hold about the system.

Thus, after we name the variables of the system such that $SI(Y:\{X;Z\})\leq SI(X:\{Y;Z\}) \leq SI(Z:\{X;Y\})$, we express the coarser PID atoms in Eq.~\ref{eq:entropydecfine} in terms of the minimal set to obtain:
\begin{align}\label{eq:entropydecfinal}
H(X,Y,Z) -  H_{(1)} & = RSI(Y \overset{\text{Z}}{\leftrightarrow} X)+ 2 \hspace{1mm} RCI(Y \overset{\text{Z}}{\leftrightarrow} X)+ \\
& + RUI(X \overset{\text{Z}}{\leftrightarrow} Y) + RUI(Y \overset{\text{X}}{\leftrightarrow} Z) + RUI(X \overset{\text{Y}}{\leftrightarrow} Z) + \nonumber \\
& + 2 \hspace{1mm} IRSI(Z \overset{\text{Y}}{\leftarrow} X) + 3 \hspace{1mm} IRSI(X \overset{\text{Z}}{\leftarrow} Y).\nonumber 
\end{align}
Unlike Eq.~\ref{eq:entropydecfine}, the entropy decomposition expressed in Eq.~\ref{eq:entropydecfinal} and illustrated in Fig.~\ref{fig:dec_entropy_final} fully describes the distribution of information in trivariate systems without the need of a specific perspective about the system. Importantly, this decomposition is unique: even though the co-information can be expressed in different ways in terms of conditional and unconditional mutual informations, in terms of the subatoms it is uniquely represented as the orange block minus the gray block (see Fig.~\ref{fig:dec_entropy_final}). Similarly, the conditional mutual information terms of Eq.~\ref{eq:dtc} are composed by the same blocks independently of the PID, as highlighted in Fig.~\ref{fig:entropyconst}.

\subsection{Describing $H(X,Y,Z)$ for dyadic and triadic systems}

To test the usefulness of the finer entropy decomposition in Eq.~\ref{eq:entropydecfinal}, we now compute its terms for the dyadic and the triadic dependencies considered in Ref.~\cite{James2016} and defined in Figure~\ref{fig:dyadic}. In both cases $H_{(1)}=0$.
For the dyadic system, there are only three positive quantities in the minimal set: the three reversible unique informations $RUI(X \overset{\text{Y}}{\leftrightarrow} Z)=RUI(Y \overset{\text{X}}{\leftrightarrow} Z)=RUI(X \overset{\text{Z}}{\leftrightarrow} Y)=1$ bit. For the triadic system, there are only two positive quantities in the minimal set: $RSI(X \overset{\text{Z}}{\leftrightarrow} Y)=RCI(X \overset{\text{Z}}{\leftrightarrow} Y)=1$ bit, but $RCI(X \overset{\text{Z}}{\leftrightarrow} Y)$ is counted twice in the $DTC$. We illustrate the resulting entropy decompositions, according to Eq.~\ref{eq:entropydecfinal}, in Fig.~\ref{fig:dec_entropy_dyadic}.
\begin{figure}[t!]
\centering
\includegraphics[width=.7\textwidth]{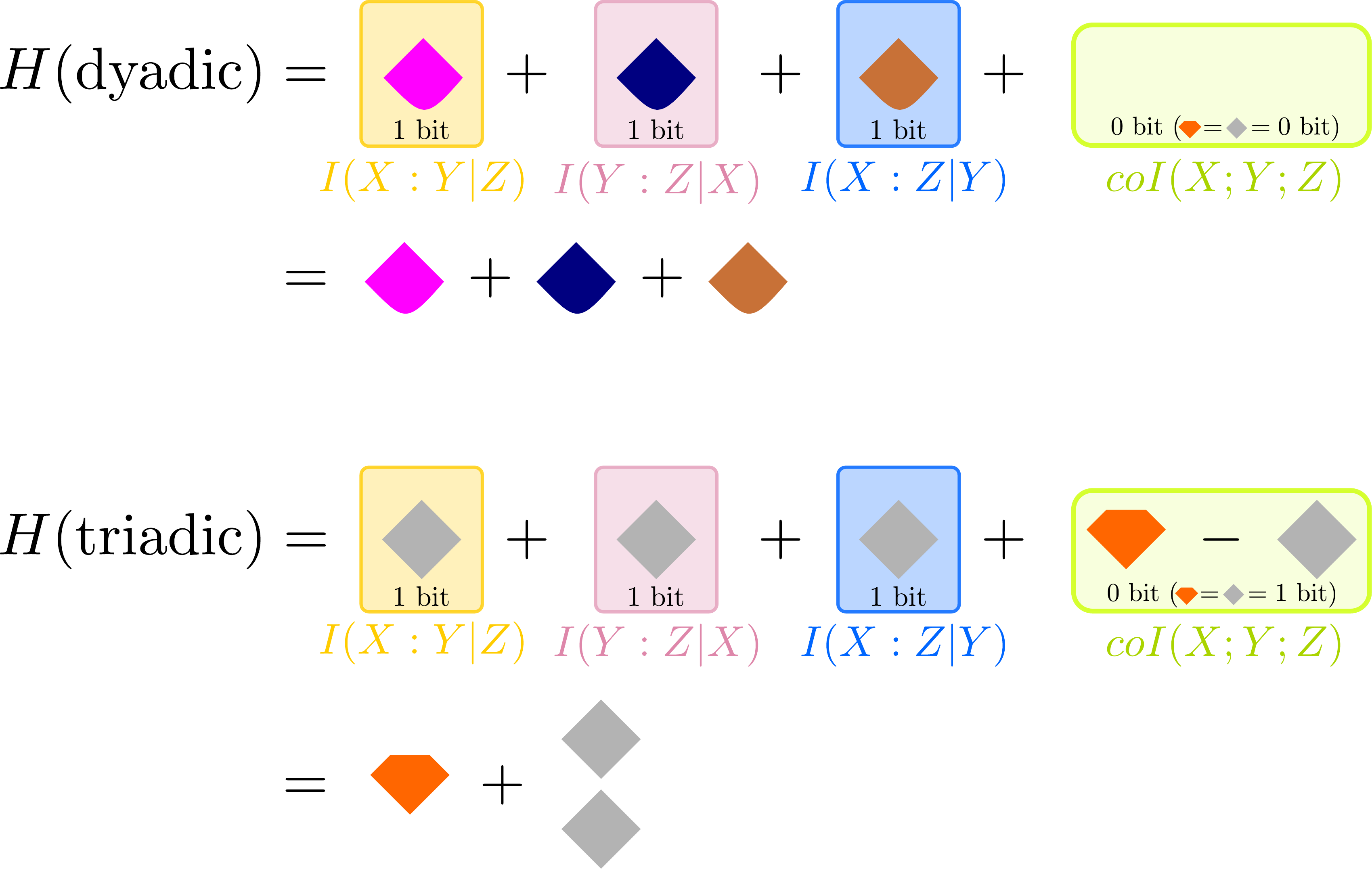}
\caption{The joint entropy $H(X,Y,Z)=3$ bit of a dyadic (upper panel) and a triadic (lower panel) system, as defined in Fig.~\ref{fig:dyadic}, is decomposed in terms of the minimal set as illustrated in Fig.~\ref{fig:dec_entropy_final}. In the dyadic system, $H(X,Y,Z)$ is decomposed into three pieces of (reversible) unique information: each variable contains 1 bit of unique information with respect to the second variable about the third variable. In the triadic system, $H(X,Y,Z)$ is decomposed into one bit of information shared among all three variables (the reversible redundancies) and two bits of (reversible) synergistic information due to the three-wise XOR structure. 
}
\label{fig:dec_entropy_dyadic}
\end{figure}

The decompositions in Fig.~\ref{fig:dec_entropy_dyadic} enable a clear interpretation of how information is finely distributed within dyadic and triadic dependencies.
The three bits of the total $H(X,Y,Z)$ in the dyadic system are seen to be distributed equally among unique information modes: each variable contains 1 bit of unique information with respect to the second variable about the third variable. Further, these unique information terms are all reversible, which reflects the symmetry of the system under pairwise swapping of the variables. This description provides a simple and accurate summary of the total entropy of the dyadic system, which matches the dependency structure illustrated in Fig.~\ref{tab:dyadic}.

The three bits of the total $H(X,Y,Z)$ in the triadic system consist of one bit of the smallest
reversible redundancy and two bit of the smallest
reversible synergy, since the latter appears twice in $H(X,Y,Z)$. Again, the reversible nature of these pieces of information reflects the symmetry of the system under pairwise swapping of the variables. Further, the bit of reversible redundancy represents the bit of information that is redundantly available to all three variables, while the two bit of reversible synergy are due to the three-wise XOR structure (see Fig.~\ref{tab:triadic}).
Why does the XOR structure provide two bits of synergistic information? Because if $X=Y \oplus Z$ then the only positive quantity in the set of subatoms in Fig.~\ref{fig:entropyconst}b is the smallest $RCI$, which however appears twice in the entropy $H(X,Y,Z)$. Importantly, these two bits of synergy do not come from the same PID diagram: our entropy decomposition in Eq.~\ref{eq:entropydecfinal} could account for both bits only because it fundamentally relies on cross-comparisons between different PID diagrams, as illustrated in Fig.~\ref{fig:entropyconst}.

\section{Applications of the finer structure of the PID framework}\label{sec:characterizing}

The aim of this Section is to show the additional insights that the finer structure of the PID framework, unveiled in Section ~\ref{sec:structure} and Fig.~\ref{fig:entropyconst}, can bring to the analysis of trivariate systems. We examine paradigmatic examples of trivariate systems and calculate the novel PID quantities of source and non-source redundancy that we described in Sec.~\ref{sec:sourcered}. Most of these examples have been considered in the literature ~\cite{Harder2013,Griffith2014,Bertschinger2014,Barrett2015,Rosas2016} to validate the definitions, or to suggest interpretations, of the PID atoms. We also discuss how $SR$ matches the notion of source redundancy introduced in Ref.~\cite{Harder2013} and discussed in Ref.~\cite{Wibral2017}. Finally, we suggest and motivate a practical interpretation of the reversible redundancy subatom $RSI$.

Even though our definitions of the minimal subatoms' set in Fig.~\ref{fig:entropyconst} only rely on Williams and Beer's axioms, some of the examples below will require a specific definition of the PID atoms that goes beyond those axioms. In those cases, our computations rely on the definitions of PID that were proposed in Ref.~\cite{Bertschinger2014}, which are the most widely accepted in the literature for trivariate systems. Where numerical computations of the PID atoms are involved, they have been performed with a software package that will be publicly released upon publication of the present work.

\subsection{Computing source and non-source redundancy}
\subsubsection{Copying --- the redundancy arises entirely from source correlations} \label{sec:copy}
\begin{figure}[t!]
\centering
\begin{subfigure}{.49\linewidth}
\begin{overpic}[trim={0cm 1cm 1cm 0cm},clip,width=\linewidth]{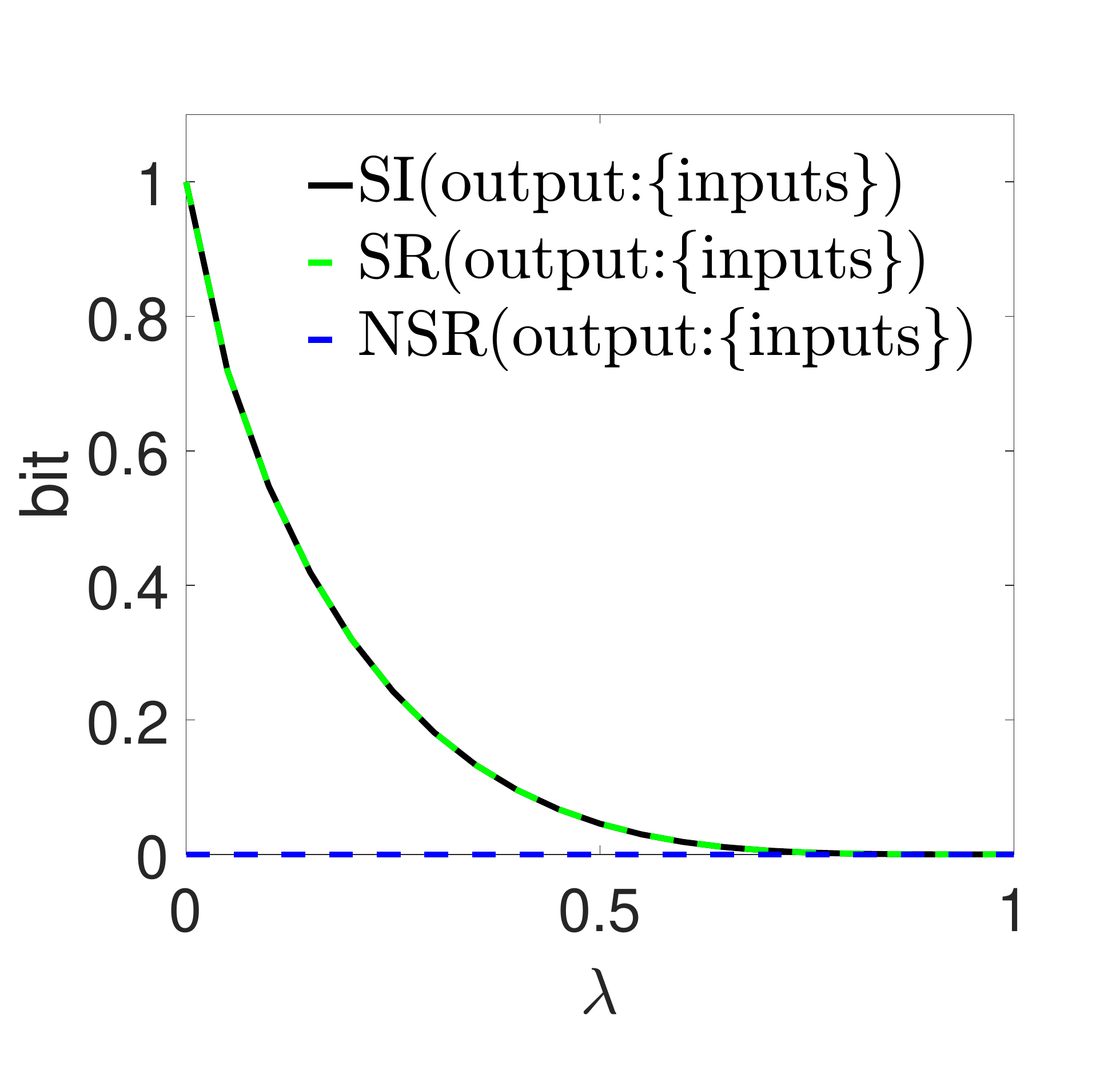}
\end{overpic}
\caption{}
\label{fig:copy}
\end{subfigure}
\begin{subfigure}{.49\linewidth}
\begin{overpic}[trim={0cm 1cm 1cm 0cm},clip,width=\linewidth]{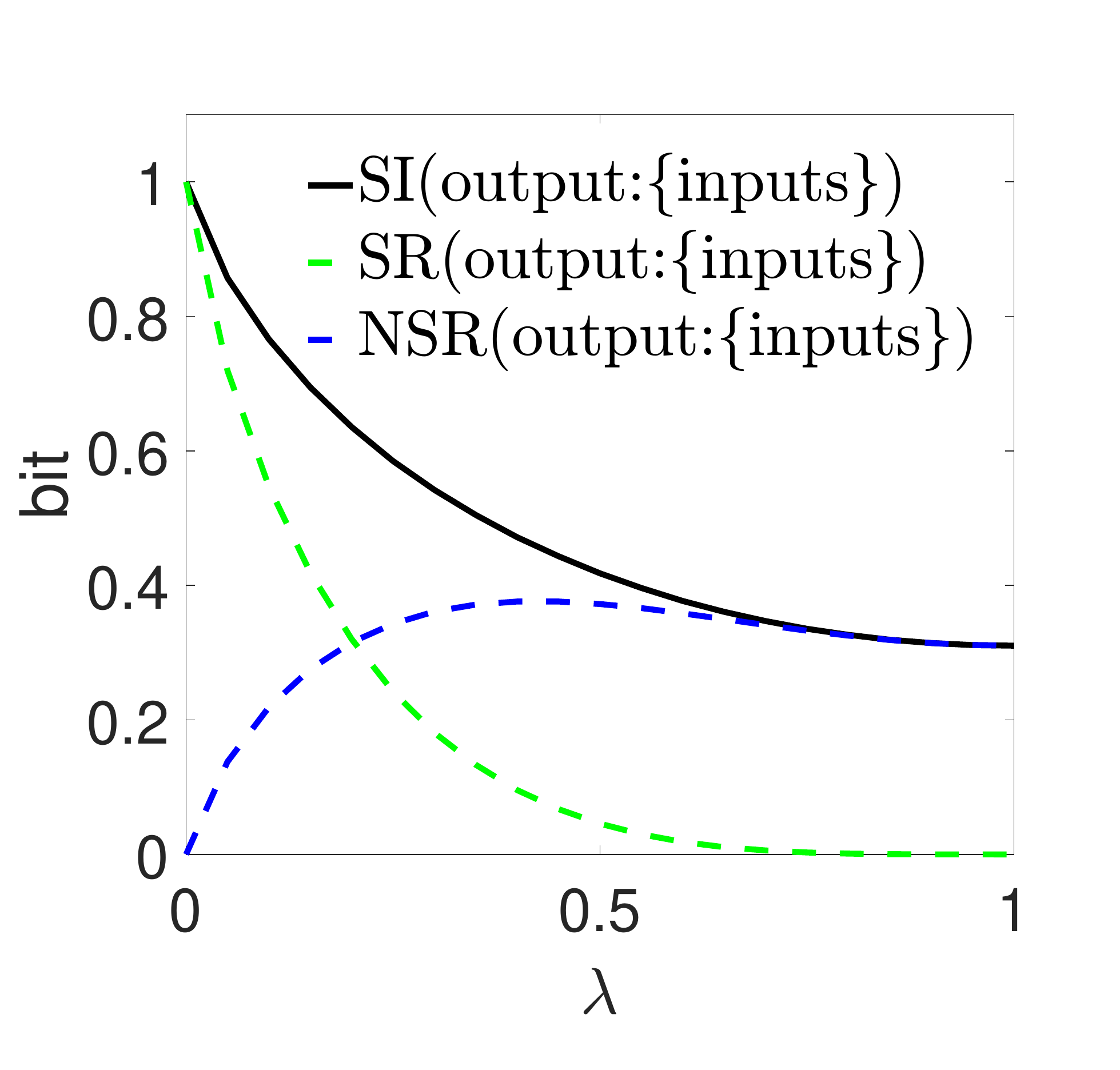}
\end{overpic}
\caption{}
\label{fig:and}
\end{subfigure}
\caption{The binary random variables $Y$ and $Z$ are uniformly distributed inputs that determine the output $X$ as $X=(Y,Z)$ in panel \textbf{a)} and as $X=Y \wedge Z$ in panel \textbf{b)}. Correlations between $Y$ and $Z$ decrease with increasing $\lambda$, from perfect correlation ($\lambda=0$) to perfect independence ($\lambda=1$). The full redundancy $SI(X:\{Y;Z\})$, the source redundancy $SR(X:\{Y;Z\})$ and the non-source redundancy $NSR(X:\{Y;Z\})$ of the inputs about the output are plotted as a function of $\lambda$. \textbf{a} Since the output variable $X$ copies the inputs $(Y, Z)$, all of $SI(X:\{Y;Z\})$ can only come from the correlations between the inputs, which is reflected in $NSR(X:\{Y;Z\})$ being identically 0 for all values of $\lambda$. \textbf{b} For all values of $\lambda>0$ we find $NSR(X:\{Y;Z\})>0$, i.e. there is a part of the redundancy $SI(X:\{Y;Z\})$ which does not arise from the correlations between the inputs $Y$ and $Z$. Accordingly, $NSR(X:\{Y;Z\})$ also appears in the synergy $CI(X:\{Y;Z\})$. }
\end{figure}
Consider a system where $Y$ and $Z$ are random binary variables that are correlated according to a control parameter $\lambda$~\cite{Harder2013}. For example, consider a uniform binary random variable $W$ that 'drives' both $Y$ and $Z$ with the same strength $\lambda$. More precisely, $p(y|w)=\lambda/2+(1-\lambda) \delta_{yw}$ and $p(z|w)=\lambda/2+(1-\lambda) \delta_{zw}$~\cite{Harder2013}. This system is completed by taking $X=(Y,Z)$, i.e. a two-bit random variable that reproduces faithfully the joint outcomes of the generating variables $Y, Z$.

We consider the inputs $Y$ and $Z$ as the PID sources and the output $X$ as the PID target, thus selecting the left-most PID diagram in Fig.~\ref{fig:PIDs}. Fig.~\ref{fig:copy} shows our calculations of the full redundancy $SI(X: \{Y;Z\})$, the source redundancy $SR(X: \{Y;Z\})$ and the non-source redundancy $NSR(X: \{Y;Z\})$, based on the definitions in Ref.~\cite{Bertschinger2014}. The parameter $\lambda$ is varied between $\lambda=0$, corresponding to $Y=Z$, and $\lambda=1$, corresponding to $Y \independent Z$. Since $NSR(X: \{Y;Z\})=0$ for any $0\leq \lambda\leq 1$, we interpret that all the redundancy $SI(X: \{Y;Z\})$ arises from the correlations between $Y$ and $Z$ (which are tuned with $\lambda$). This is indeed compatible with the discussion in Ref.~\cite{Harder2013}, where the authors argued that in the 'copying' example the entire redundancy should be already apparent in the sources.

\subsubsection{AND gate: the redundancy is not entirely related to source correlations}

Consider a system where the correlations between two binary random variables, the inputs $Y$ and $Z$, are described by the control parameter $\lambda$ as in~\ref{sec:copy}, but the output $X$ is determined by the $AND$ function as $X=Y \wedge Z$~\cite{Harder2013}. As the causal structure of the system would suggest, we consider the inputs $Y$ and $Z$ as the PID sources and the output $X$ as the PID target, thus selecting the left-most PID diagram in Fig.~\ref{fig:PIDs}. Fig.~\ref{fig:and} shows our calculations of the full redundancy $SI(X: \{Y;Z\})$, the source redundancy $SR(X: \{Y;Z\})$ and the non-source redundancy $NSR(X: \{Y;Z\})$, based on the definitions in Ref.~\cite{Bertschinger2014}.

$SR$ and $NSR$ now show a non-trivial behavior as a function of the $Y-Z$ correlation parameter $\lambda$. If $\lambda=0$ and thus $Y=Z$, the full redundancy is made up entirely of source redundancy --- trivially, both $SI(X: \{Y;Z\})$ and $SR(X: \{Y;Z\})$ equal the mutual information $I(Y:Z)=H(Y)=1$ bit. When $\lambda$ increases, the full redundancy decreases monotonically to its minimum value of $\approx 0.311$ bit for $\lambda=1$ (when $Y \independent Z$). 
Importantly, $SR(X: \{Y;Z\})$ decreases monotonically as a function of $\lambda$ to its minimum value of zero bit when $\lambda=1$: this behavior is indeed expected from a measure that quantifies correlations between the sources that also produce redundant information about the target (see Section~\ref{sec:sourcered}). On the other hand, $NSR(X: \{Y;Z\})>0$ for $\lambda>0$, and it increases as a function of $\lambda$. When $\lambda=1$, i.e. when $Y \independent Z$, $NSR$ corresponds to the full redundancy. This is compatible with our description of non-source redundancy (see Section ~\ref{sec:sourcered}) as redundancy that is not related to the source correlations; indeed, $NSR>0$ implies that the sources also carry synergistic information about the target (here, due to the relationship $X=Y \wedge Z$). $NSR$ thus also quantifies the notion of mechanistic redundancy that was introduced in Ref.~\cite{Harder2013} with reference to this scenario.

\subsubsection{ Dice sum: tuning irreversible redundancy}\label{section:dice}
\begin{figure}[t!]
\begin{subfigure}{.49\textwidth}
\begin{overpic}[trim={0 1cm 1cm 0cm},clip,width=\linewidth]{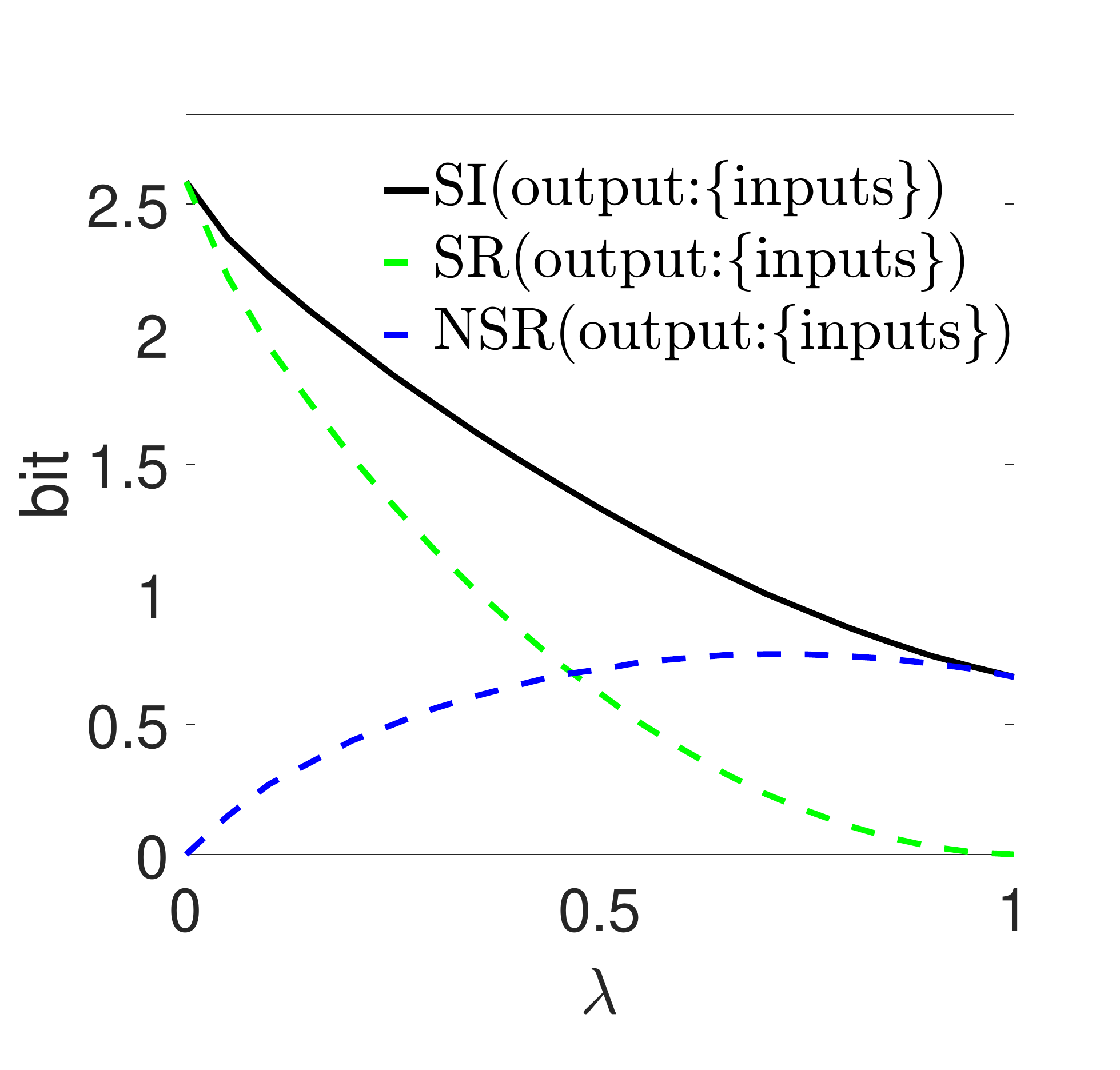}
\end{overpic}
\caption{}
\label{fig:dice_suma}
\end{subfigure}
\begin{subfigure}{.49\textwidth}
\begin{overpic}[trim={0 1cm 1cm 0cm},clip,width=\linewidth]{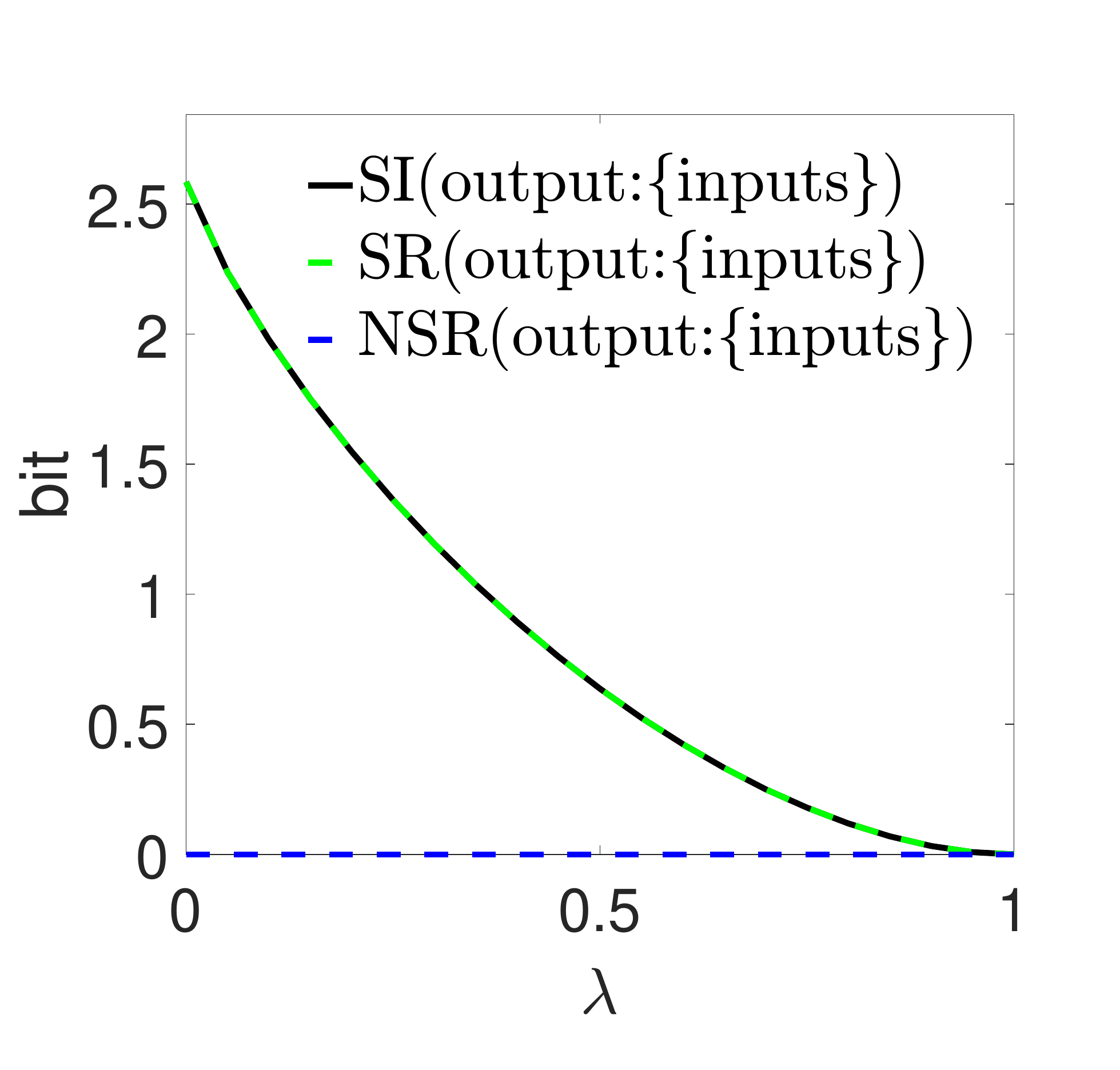}
\end{overpic}
\caption{}
\label{fig:dice_sumb}
\end{subfigure}
\caption{Two dice with uniformly distributed outcomes $(y, z)$ ranging from 1 to 6 are the inputs fed to an output third variable $X$ as follows: \textbf{a)} $x=y+z$; \textbf{b)} $x=y+6$ $z$. The parameter $\lambda$ controls the correlations between $Y$ and $Z$ (from complete correlation for $\lambda=0$ to complete independence for $\lambda=1$). In panel \textbf{a)} the inputs $Y, Z$ are symmetrically combined to determine the output $X$, while in panel \textbf{b)} this symmetry is lost: accordingly, for any fixed $\lambda$ the full redundancy $SI(X:\{Y;Z\})$ is larger in \textbf{a)} than in \textbf{b)}. Since the value of $\alpha$ does not change the inputs' correlations, the relative contribution of the inputs' correlations to the full redundancy $SI(X:\{Y;Z\})$ increases from \textbf{a)} to \textbf{b)}. Indeed,
in panel \textbf{a)} $NSR(X:\{Y;Z\})>0$ for $\lambda>0$, while in \textbf{b)} $SI(X:\{Y;Z\})=SR(X:\{Y;Z\})$ for any $\lambda$, i.e. the redundancy arises entirely from inputs' correlations.}
\label{fig:dice_sum}
\end{figure}
Consider a system where $Y$ and $Z$ are two uniform random variables, each representing the outcome of a die throw~\cite{Harder2013}. A parameter $\lambda$ controls the correlations between the two dice: $p(y,z)=\lambda/36 + (1-\lambda)/6 \hspace{1mm} \delta_{yz}$. Thus, for $\lambda=0$ the dice throws always match, for $\lambda=1$ the outcomes are completely independent. Further, the output $X$ combines each pair $(y,z)$ of input outcomes as $x=y + \alpha$ $z$, where $\alpha \in \{1,2,3,4,5,6\}$. This example was suggested in Ref.~\cite{Harder2013} specifically to point out the conceptual difficulties in quantifying the interplay between the redundancy and the source correlations, that we addressed with the identification of $SR$ and $NSR$ (see Section ~\ref{sec:sourcered}).

Harder \emph{et al.} calculated redundancy by using their own proposed measure $I_{\text{red}}$, which also abides by the PID axioms in Ref.~\cite{Williams2010} but is different than the measure $SI(X:\{Y;Z\})$ introduced in Ref.~\cite{Bertschinger2014}. However, Ref.~\cite{Bertschinger2014} also calculated $SI(X:\{Y;Z\})$ for this example: they showed that the differences between $SI(X:\{Y;Z\})$ and  $I_{\text{red}}$ in this example are only quantitative (and there is no difference at all for some values of $\alpha$), while the qualitative behaviour of both measures as a function of $\lambda $ is very similar.

Fig.~\ref{fig:dice_sum} shows our calculations of the full redundancy $SI(X: \{Y;Z\})$, the source redundancy $SR(X: \{Y;Z\})$ and the non-source redundancy $NSR(X: \{Y;Z\})$, based on the definitions in Ref.~\cite{Bertschinger2014}. We display the two 'extreme' cases $\alpha=1$, $\alpha=6$ for illustration.

With $\alpha=6$, $X$ is isomorphic to the joint variable $(Y,Z)$: for any $\lambda$, this implies on one hand $SI(X:\{Y;Z\})=SI((Y;Z):\{Y;Z\})=I(Y:Z)$, and on the other $UI(Y:\{Z\backslash X\})=0$ and thus $SI(Y:\{X;Z\})=I(Y:Z)$ (see Eq.~\ref{eq:pid1}).
According to Eq.~\ref{eq:SR2}, we thus find that the source redundancy $SR(X:\{Y;Z\})$ saturates, for any $\lambda$, the general inequality in Eq.~\ref{eq:srbound}: all correlations between the inputs $Y, Z$ also produce redundant information about $X$ (see Section ~\ref{sec:sourcered}). Further, according to Eq.~\ref{eq:NSR}, we find $NSR(X:\{Y;Z\})=0$ for any $\lambda$ (see Fig.~\ref{fig:dice_sumb}): we thus interpret that all the redundancy $SI(X:\{Y;Z\})$ arises from correlations between the inputs.
Instead, if we fix $\lambda$ and decrease $\alpha$, the two inputs $Y, Z$ are more and more symmetrically combined in the output $X$: with $\alpha=1$, the pieces of information respectively carried by each input about the output overlap maximally. Correspondingly, the full redundancy $SI(X:\{Y;Z\})$ increases ~\cite{Harder2013}. However, keeping $\lambda$ fixed does not change the inputs' correlations. Thus, we expect that the relative contribution of the inputs' correlations to the full redundancy $SI(X:\{Y;Z\})$ should decrease proportionally. Indeed, in Fig.~\ref{fig:dice_suma} we find $NSR(X:\{Y;Z\})>0$ for $\lambda>0$, which signals that a part of $SI(X:\{Y;Z\})$ is not related to the inputs' correlations.

We finally note that also in this paradigmatic example the splitting of the redundancy into $SR$ and $NSR$ addresses the challenge of separating the two kinds of redundancy outlined in Ref.~\cite{Harder2013}.

\subsubsection{Trivariate jointly Gaussian systems}\label{sec:gauss}

Barrett considered in detail the application of the PID to trivariate jointly Gaussian systems $(X,Y,Z)$ in which the target is univariate~\cite{Barrett2015}: he showed that several specific proposals for calculating the PID atoms all converge, in this case, to the same following measure of redundancy:
\begin{equation}\label{eq:redgauss}
SI(X: \{Y;Z\})=\min[I(X:Y), I(X:Z)].
\end{equation}
We note that Eq.~\ref{eq:redgauss} highlights the interesting property that, in trivariate Gaussian systems with a univariate target, the redundancy is as large as it can be, since it saturates the general inequalities $SI(X: \{Y;Z\})\leq I(X:Y), I(X:Z)$.

Direct application of our definitions in Eqs.~\ref{eq:SR2} and~\ref{eq:NSR} to such systems yields:
\begin{gather}
SR(X: \{Y;Z\}) = \min[I(X:Y), I(X:Z), I(Y:Z)], \nonumber\\
NSR(X: \{Y;Z\}) = \min[I(X:Y), I(X:Z)] - \min[I(X:Y), I(X:Z), I(Y:Z)]. \label{eq:revredgauss}
\end{gather}
Thus, we find that in these systems the source redundancy is also as large as it can be, since it also saturates the general inequalities $SR(X: \{Y;Z\}) \leq I(X:Y), I(X:Z), I(Y:Z)$ (which follow immediately from its definition in Eq.~\ref{eq:SR2}). 
Further, combining Eq.~\ref{eq:redgauss} with Eq.~\ref{eq:revredgauss} gives:
\begin{gather}
SR(X: \{Y;Z\}) = \begin{cases}
  SI(X: \{Y:Z\}), & \text{if} \hspace{1mm}  I(Y:Z) \geq SI(X: \{Y;Z\}), \\
  I(Y:Z), & \text{otherwise};
\end{cases}\label{eq:redrevgauss}\\
NSR(X: \{Y;Z\}) = \begin{cases}
 0 , & \text{if} \hspace{1mm} I(Y:Z)\geq SI(X: \{Y;Z\}), \\
 SI(X: \{Y:Z\})-I(Y:Z), & \text{otherwise}.
\end{cases}\label{eq:redirrevgauss}
\end{gather}
This identification of source redundancy, which quantifies \emph{pairwise correlations between the sources that also produce redundant information about the target} (see Section~\ref{sec:sourcered}), provides more insight about the distribution of information in Gaussian systems. Indeed, the property that source redundancy is maximal indicates that the correlations between any pair of source variables (for example, $\{Y; Z\}$ as considered above) produces as much redundant information as possible about the corresponding target ($X$ as considered above). Accordingly, when $I(Y:Z)< SI(X: \{Y;Z\})$, the redundancy also includes some non-source redundancy $NSR(X: \{Y;Z\})>0$ that implies the existence of synergy, i.e. $CI(X:\{Y;Z\})>0$.

\subsection{$RSI$ quantifies information between two variables that also passes monotonically through the third}\label{sec:passes}

In this section we discuss a practical interpretation of the reversible redundancy
subatom $RSI$ defined in Eq.\ref{eq:rsi1}. We note that $RSI(X \overset{\text{Z}}{\leftrightarrow} Y)$ appears both in $SI(X:\{Y;Z\})$ and in $SI(Y:\{X;Z\})$, i.e. it quantifies a common amount of information that $Z$ shares with each of the two other variables about the third variable. We shorten this description into the statement that $RSI(X \overset{\text{Z}}{\leftrightarrow} Y)$ quantifies 'information between $X$ and $Y$ that also passes through $Z$'. Reversible redundancy is further characterized by the following Proposition:
\begin{Proposition}\label{eq:monotone}
$RSI(X \overset{\text{(Z,Z')}}{\leftrightarrow} Y)$ $\geq$ $RSI(X \overset{\text{Z}}{\leftrightarrow} Y)$.
\end{Proposition}
\begin{proof}\hspace{1mm}
Both $SI(X:\{Y;(Z,Z')\})\geq SI(X:\{Y;Z\})$ and $SI(Y:\{X;(Z,Z')\})\geq SI(Y:\{X;Z\})$ ~\cite{Bertschinger2013}. Thus, $RSI(X \overset{\text{Z}}{\leftrightarrow} Y)=\min[ SI(X:\{Y;Z\}), SI(Y:\{X;Z\})] \leq  \min [SI(X:\{Y;(Z,Z')\}), SI(Y:\{X;(Z,Z')\})]= RSI(X \overset{\text{(Z,Z')}}{\leftrightarrow} Y) $.
\end{proof}
Indeed, the fact that $RSI$ always increases whenever we expand the middle variable corresponds to the increased capacity of the entropy of the middle variable to host information between the endpoint variables.
We discuss this interpretation of $RSI$ by examining several examples, where we can motivate \emph{a priori} our expectations about this novel mode of information sharing. 

\subsubsection{Markov chains}

Consider the most generic Markov chain $X \rightarrow Z \rightarrow Y$, defined by the property that $p(x,y|z)=p(x|z)p(y|z)$, i.e. that $X$ and $Y$ are conditionally independent given $Z$.
The Markov structure allows us to formulate a clear \emph{a priori} expectation about the amount of information between the endpoints of the chain, $X$ and $Y$, that also 'passes through the middle variable $Z$': this information should clearly equal $I(X:Y)$, because whatever information is established between $X$ and $Y$ \emph{must} pass through $Z$.
\begin{figure}[t!]
\centering
\includegraphics[width=.7\textwidth]{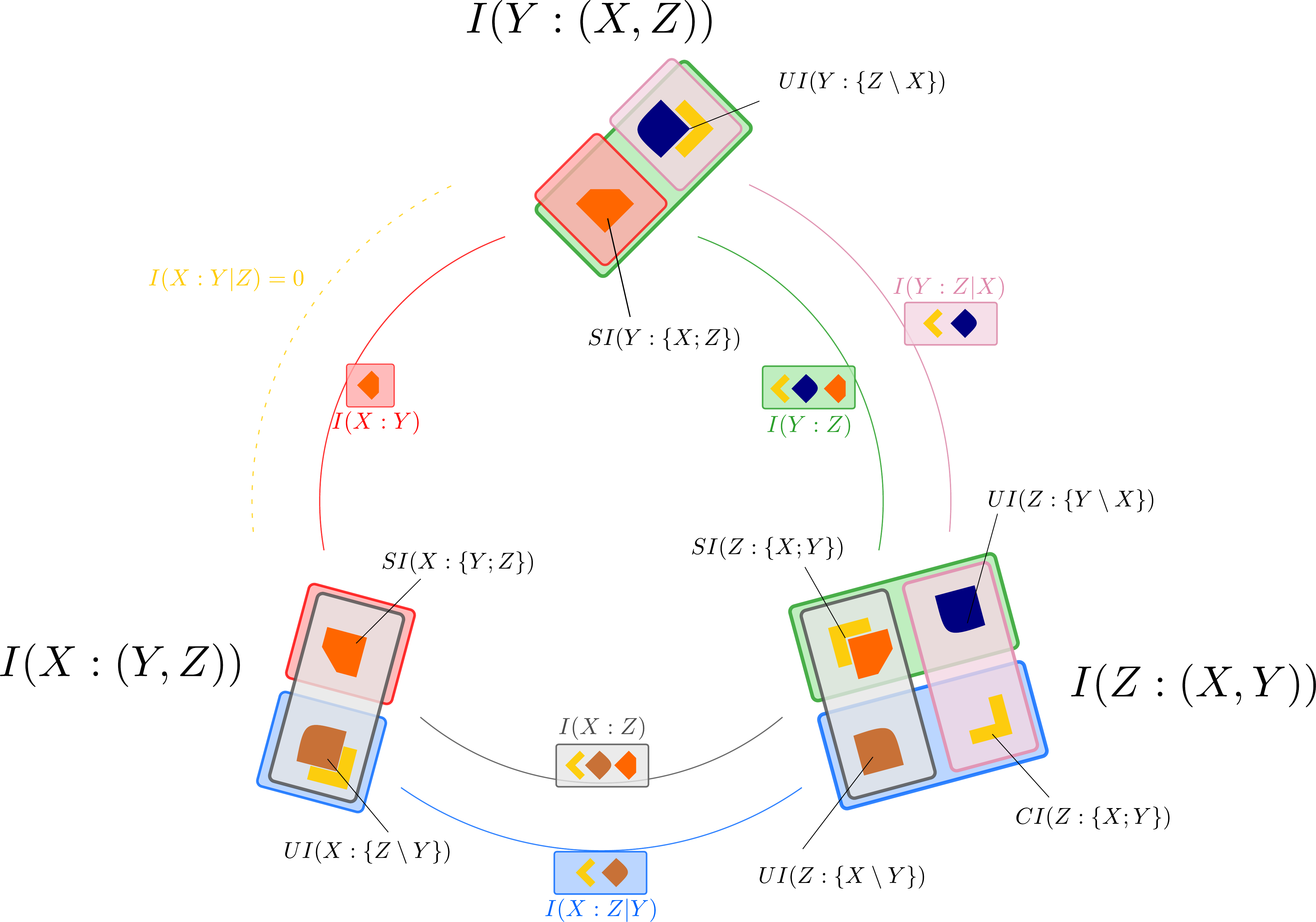}
\caption{The finer structure of the three PID diagrams, as defined in Fig.~\ref{fig:entropyconst}, in the case of a generic Markov chain $X \rightarrow Z \rightarrow Y$. Three of the seven subatoms of the minimal set described in Fig.~\ref{fig:entropyconst} are forced to zero by the Markov property $I(X:Y|Z)=0$. }
\label{fig:markovdec}
\end{figure}
Indeed,
the Markov property $I(X:Y | Z)=0$ implies, as we see immediately from Fig.~\ref{fig:PIDs}, that $SI(X:\{Y;Z\})=SI(Y:\{X;Z\})=I(X:Y)$. By virtue of Eqs.~\ref{eq:rsi1} and~\ref{eq:irsi1}, in accordance with our expectation, we find
\begin{gather}
RSI(X \overset{\text{Z}}{\leftrightarrow} Y)=I(X:Y),\label{eq:revmarkov} %\\
\end{gather}
which holds true independently of the marginal distributions of $X$, $Y$ and $Z$. Notably, the symmetry of $RSI(X \overset{\text{Z}}{\leftrightarrow} Y)$ under swap of the endpoint variables $X \leftrightarrow Y$ is also compatible with the property that $X \rightarrow Z \rightarrow Y$ is a Markov chain if and only if $Y \rightarrow Z \rightarrow X$ is a Markov chain. In words, whatever information flows from $X$ through $Z$ to $Y$ equals the information flowing from $Y$ through $Z$ to $X$.

More generally, the Markov property $I(X:Y | Z)=0$ implies that $RCI(Y \overset{\text{Z}}{\leftrightarrow} X)=IRSI(X \overset{\text{Z}}{\leftarrow} Y)=RUI(X \overset{\text{Z}}{\leftrightarrow} Y)=0$: thus, only four of the seven subatoms of the minimal set in Fig.~\ref{fig:entropyconst} can be larger than zero. The three PIDs of the system, decomposed with the minimal set, are shown in Fig.~\ref{fig:markovdec}. In particular, we see that also $RSI(X \overset{\text{Y}}{\leftrightarrow} Z)=I(X:Y)$, thus matching our expectation that in the Markov chain $X \rightarrow Z \rightarrow Y$ the information between $X$ and $Z$ that also passes through $Y$ still equals $I(X:Y)$. 
We finally note that none of the results regarding Markov chains depends on specific definitions of the PID atoms: they were derived only on the basis of the PID axioms in Ref.~\cite{Williams2010}.

\subsubsection{Two parallel communication channels }
Consider five binary uniform random variables $X_{1}, X_{2}, Y_{1}, Y_{2}, Z$ with three parameters $0 \leq \lambda_1, \lambda_2, \lambda_3 \leq 1$ controlling the correlations $X_{1} \overset{\lambda_1}{\leftrightarrow} Z \overset{\lambda_2}{\leftrightarrow} Y_{1}$, $X_{2} \overset{\lambda_3}{\leftrightarrow} Y_{2}$ (in the same way $\lambda$ controls the $Y \overset{\lambda}{\leftrightarrow} Z $ correlations in ~\ref{sec:copy}). We consider the trivariate system ($X, Y, Z$) with $X=(X_{1},X_{2})$ and $Y=(Y_{1},Y_{2})$ (see Fig.~\ref{fig:parallel}). We intuitively expect that the information between $X$ and $Y$ that also passes through $Z$, in this case, should equal $I(X_{1}:Y_{1})$, which in general will be smaller than $I(X:Y)=I(X_{1}:Y_{1})+I(X_{2}:Y_{2})$. Indeed, we computed the PID atoms with the definitions in Ref.~\cite{Bertschinger2014} over a fine discretization of the $(\lambda_1, \lambda_2, \lambda_3)$-space and we always found that $RSI(X \overset{\text{Y}}{\leftrightarrow} Z)=I(X_{1}:Y_{1})$.

\begin{figure}[h!]
\centering
\begin{overpic}[trim={2cm 9cm 8cm 2cm},clip,width=.8\textwidth]{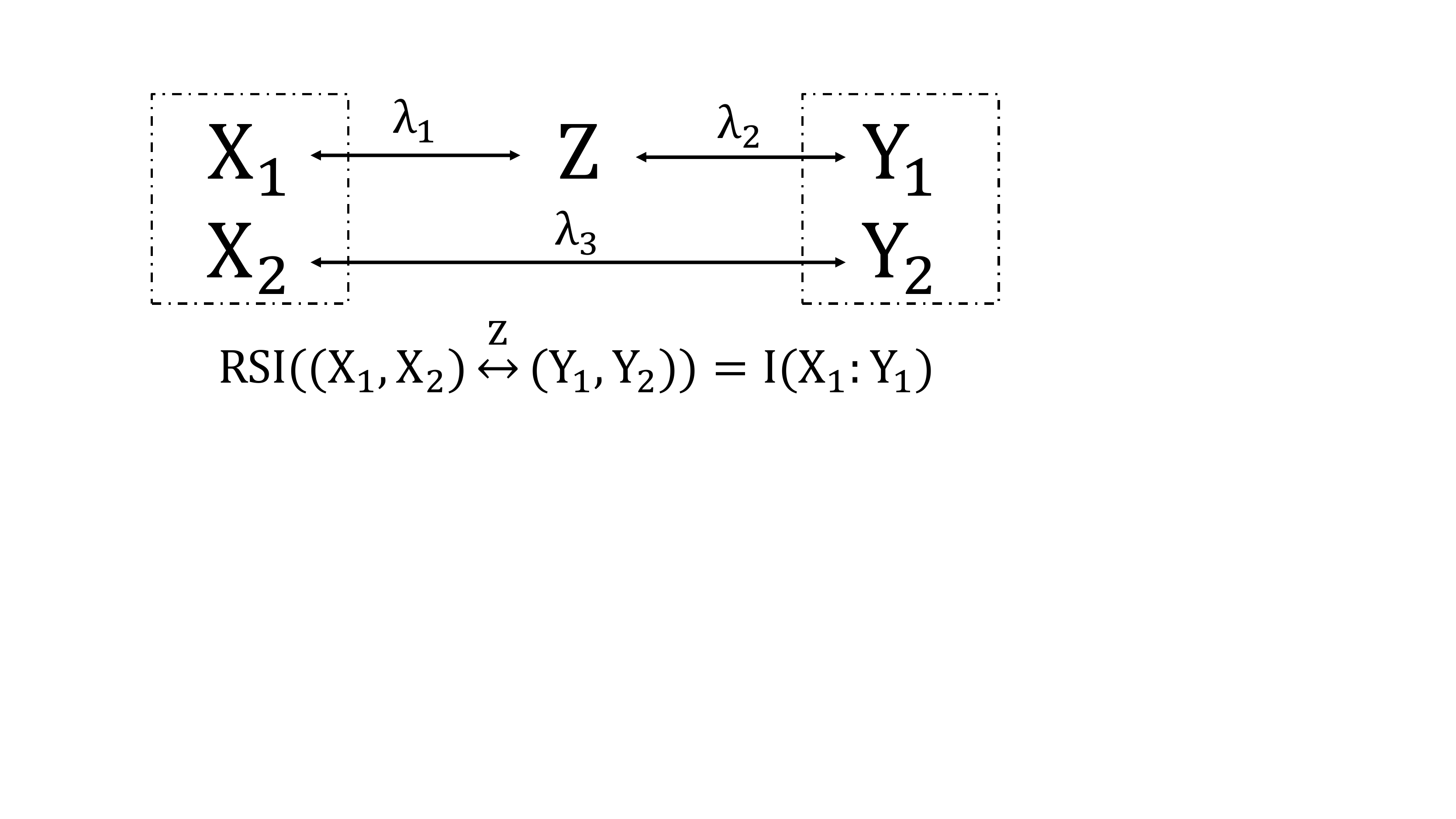}
\end{overpic}
\caption{$X=(X_{1},X_{2})$ and $Y=(Y_{1},Y_{2})$ share information via two parallel channels: one passes through $Z$, the other does not. The parameters $\lambda_1, \lambda_2, \lambda_3$ control the correlations as depicted with the arrows. In agreement with the interpretation in Section~\ref{sec:passes}, $RSI(X \overset{\text{Z}}{\leftrightarrow} Y)$ here quantifies information between $X$ and $Y$ that also passes through $Z$, which clearly equals $I(X_{1}:Y_{1})$.}
\label{fig:parallel}
\end{figure}

\subsubsection{Other examples}

To further describe the interpretation of $RSI(X \overset{\text{Z}}{\leftrightarrow} Y) $ as information between $X$ and $Y$ that also passes through $Z$, we here reconsider the examples, among those discussed before Section~\ref{sec:passes}, where we can formulate intuitive expectations about this information sharing mode.
In the dyadic system described in Fig.~\ref{tab:dyadic}, our expectation is that there should be no information between two variables that also passes through the other: indeed, $RSI(X \overset{\text{Z}}{\leftrightarrow} Y) = 0$ in this case. Instead, in the triadic system described in Fig.~\ref{tab:triadic}, we expect that the information between two variables that also passes through the other should equal the information that is shared among all three variables, which amounts to 1 bit. Indeed, we find $RSI(X \overset{\text{Z}}{\leftrightarrow} Y) = 1$ bit.
In the 'copying' example in Section ~\ref{sec:copy} $X=(Y,Z)$, thus we expect that $I(Y:Z)$ corresponds to information between $Y, Z$ that also passes through $X$, but also to information between $X$ and $Z$ that also passes through $Y$. Indeed, $RSI(X \overset{\text{Z}}{\leftrightarrow} Y) = RSI(X \overset{\text{Y}}{\leftrightarrow} Z) = I(Y:Z)$. We remark that all the values of $RSI$ in these examples depend on the specific definitions of the PID atoms in Ref.~\cite{Bertschinger2014}.

\section{Discussion}

%This section may be divided by subheadings. Authors should discuss the results and how they can be interpreted in perspective of previous studies and of the working hypotheses. The findings and their implications should be discussed in the broadest context possible. Future research directions may also be highlighted.

The Partial Information Decomposition (PID) pioneered by Williams and Beer has provided an elegant construction to quantify, with information theory, different ways two stochastic variables can carry information about a third variable~\cite{Williams2010}. In particular, it has enabled consistent definitions of synergy, redundancy and unique information components among three stochastic variables~\cite{Bertschinger2014}. More generally, it has generated considerable interest as it addresses the difficult yet practically important problem of extending information theory beyond the classic bivariate measures of Shannon to  fully characterize multivariate dependencies. 

However, the axiomatic PID construction, as originally formulated by Williams and Beer, fundamentally relied on the possibly arbitrary classification of the three variables into sources and target, and this partial perspective prevented a complete and general information-theoretic description of trivariate systems. More specifically, the original PID framework could not quantify some important modes of information sharing, such as source redundancy~\cite{Harder2013}, and could not allot the full joint entropy of some important trivariate systems~\cite{James2016}. The work presented here addresses these issues by extending the original PID framework in two respects. First, we decomposed the original PID atoms in terms of finer information subatoms with a well defined interpretation that is invariant to different variables' classifications. Then, we constructed an extended framework to completely decompose the distribution of information within any trivariate system. 

Importantly, our formulation did not require the addition of further axioms to the original PID construction. We proposed that distinct PIDs for the same system, corresponding to different target selections, should be evaluated and then compared to identify how the decomposition of information changes across different perspectives. More specifically, we identified reversible pieces of information ($RSI$, $RUI$, $RCI$) that contribute to the same kind of PID atom if we reverse the roles of target and source between two variables. The complementary subatomic components of the PID lattices are the irreversible pieces of information ($IRSI$), that contribute to different kinds of PID atom for different target selections. These subatoms thus measure asymmetries between different decompositions of Shannon quantities pertaining to two different PIDs of the same system, and such asymmetries reveal the additional detail with which the PID atoms assess trivariate dependencies as compared to the coarser and more symmetric Shannon information quantities.

The crucial result of this approach was unveiling the finer structure underlying the PID lattices: we showed that an invariant minimal set of seven information subatoms is sufficient to decompose the three PIDs of any trivariate system. In the remainder of this section, possible uses of these subatoms and their implications for the understanding of systems of three variables are discussed. 

\subsection{Use of the subatoms to map the distribution of information in trivariate systems}

Our minimal subatoms' set was first used to characterize more finely the distribution of information among three variables. We clarified the interplay between the redundant information shared between two variables $A$ and $B$ about a third variable $C$, on one side, and the correlations between $A$ and $B$, on the other. We decomposed the redundancy into the sum of source-redundancy $SR$, which \emph{quantifies the part of the redundancy which arises from the pairwise correlations between the sources $A$ and $B$}, and non-source redundancy $NSR$, which can be larger than zero even if the sources $A$ and $B$ are statistically independent. Interestingly, we found that $NSR$ \emph{quantifies the part of the redundancy which implies that $A$ and $B$ also carry synergistic information about $C$}. The separation of these qualitatively different components of redundancy promises to be useful in the analysis of any complex system where several inputs are combined to produce an output~\cite{Wibral2017}.

Then, we used our minimal subatoms' set to extend the descriptive power of the PID framework in the analysis of any trivariate system. We constructed a general, unique, and nonnegative decomposition of the joint entropy $H(X,Y,Z)$ in terms of information-theoretic components that can be clearly interpreted without arbitrary variable classifications. This construction parallels the decomposition of the bivariate entropy $H(X,Y)$ in terms of Shannon's mutual information $I(X:Y)$. We demonstrated the descriptive power of this approach by decomposing the complex distribution of information in dyadic and triadic systems, which was shown not to be possible within the original PID framework~\cite{James2016}.

We gave practical examples of how the finer structure underlying the PID atoms provides more insight into the distribution of information within important and well-studied trivariate systems. In this spirit, we put forward a practical interpretation of the reversible redundancy $RSI$, and future work will address additional interpretations of the components of the minimal subatoms' set.

\subsection{Possible extensions of the formalism to multivariate systems with many sources}

The insights that derive from our extension of the PID framework suggest that the PID lattices could also be useful to characterize the statistical dependencies in multivariate systems with more than two sources. Our approach does not rely on the adoption of specific PID measures, but only on the axiomatic construction of the PID lattice. Thus, it can be immediately extended to the multivariate case by embedding trivariate lattices within larger systems' lattices~\cite{Chicharro2017}: a further breakdown of the minimal subatoms' set could be obtained if the current system were embedded as part of a bigger system. Further, when systems with more than two sources are considered, the definition of source redundancy might be extended as to determine which subatoms of a redundancy can be explained by dependencies among the sources (for example, by replacing the mutual information between two sources with a measure of the overall dependencies among all sources, such as the \emph{total correlation} introduced in Ref.~\cite{Watanabe1960}). More generally, the idea of comparing different PID diagrams that partially decompose the same information can also be generalized to identify finer structure underlying higher-order PID lattices with different numbers of variables.
These identifications might also help addressing specific questions about the distribution of information in complex multivariate systems.

\subsection{Potential implications for systems biology and systems neuroscience}

A common problem in system biology is to characterize how the function of the whole biological system is shaped by the dependencies among its many constituent biological variables. In many cases, ranging from gene regulatory networks~\cite{Margolin2006}
 to metabolic pathways~\cite{Ludtke2008} and to systems neuroscience~\cite{Averbeck2006,Quiroga2009,Panzeri2015}, an information-theoretic decomposition of how information is distributed and processed within different parts of the system would allow a model-free characterization of these dependencies. The work discussed here can be used to shed more light on these issues by allowing to tease apart qualitatively different modes of interaction, as a first necessary step to understanding the causal structure of the observed phenomena~\cite{Pearl2009}.

In systems neuroscience, the decomposition introduced here may be important for studying specific and timely questions about neural information processing. This work can contribute to the study of neural population coding, that is the study of how the concerted activity of many neurons encodes information about ecologically relevant variables such as sensory stimuli~\cite{Shamir2014,Panzeri2015}.
In particular, a key characteristic of a neural population code is the degree to which pairwise or higher-order cross-neuron statistical dependencies are used by the brain to encode and process information, in a potentially redundant or synergistic way, across neurons~\cite{Pola2003,Schneidman2003,Latham2005} and across time~\cite{Runyan2017}. 
Our work is also of potential relevance to study another hot issue in neuroscience, that is the relevance of the information about sensory variables carried by neural activity for perception~\cite{Panzeri2017,Jazayeri2017}. This is a crucial issue to resolve the diatribe about the nature of the neural code, that is the set of symbols used by neurons to encode information and produce brain function~\cite{Gallego2017,Sharpee2017,Pitkow2017}. Addressing this problem requires mapping the information in the multivariate distribution of variables such as the stimuli presented to the subject, the neural activity elicited by the presentation of such stimuli, and the behavioral reports of the subject's perception. More specifically, it requires  characterizing the information between the presented stimulus  and the behavioral report of the perceived stimulus  that can be extracted from neural activity~\cite{Panzeri2017}. It is apparent that developing general decompositions of the information exchanged in multivariate systems, as we did here, is key to succeeding in rigorously addressing these fundamental systems-level questions. 

%%%%%%%%%%%%%%%%%%%%%%%%%%%%%%%%%%%%%%%%%%
%\section{Materials and Methods}

%This section should be divided by subheadings. Materials and Methods should be described with sufficient details to allow others to replicate and build on published results. Please note that publication of your manuscript implicates that you must make all materials, data, and protocols associated with the publication available to readers. Please disclose at the submission stage any restrictions on the availability of materials or information. New methods and protocols should be described in detail while well-established methods can be briefly described and appropriately cited.

%Research manuscripts reporting large datasets that are deposited in a publicly available database should specify where the data have been deposited and provide the relevant accession numbers. If the accession numbers have not yet been obtained at the time of submission, please state that they will be provided during review. They must be provided prior to publication.

%%%%%%%%%%%%%%%%%%%%%%%%%%%%%%%%%%%%%%%%%%
%\section{Conclusions}

%This section is not mandatory, but can be added to the manuscript if the discussion is unusually long or complex.

%%%%%%%%%%%%%%%%%%%%%%%%%%%%%%%%%%%%%%%%%%
\vspace{6pt}

%%%%%%%%%%%%%%%%%%%%%%%%%%%%%%%%%%%%%%%%%%
%% optional
%\supplementary{The following are available online at www.mdpi.com/link, Figure S1: title, Table S1: title, Video S1: title.}

%%%%%%%%%%%%%%%%%%%%%%%%%%%%%%%%%%%%%%%%%%
\section*{Acknowledgments}{We are grateful to members of Panzeri's Laboratory for useful feedback, and to P. E. Latham, A. Brovelli and C. de Mulatier for useful discussions. This research was supported by the Fondation Bertarelli.}

%%%%%%%%%%%%%%%%%%%%%%%%%%%%%%%%%%%%%%%%%%
\section*{Author contributions}{All authors conceived the research; G.P., E.P. and D.C. performed the research; G.P., E.P. and D.C. wrote a first draft of the manuscript; all authors edited and approved the final manuscript; S.P. supervised the research.}

%%%%%%%%%%%%%%%%%%%%%%%%%%%%%%%%%%%%%%%%%%
%\conflictofinterests{The authors declare no conflict of interest.}

%%%%%%%%%%%%%%%%%%%%%%%%%%%%%%%%%%%%%%%%%%
%% optional
%\abbreviations{The following abbreviations are used in this manuscript:\\

%\noindent MDPI: Multidisciplinary Digital Publishing Institute\\
%DOAJ: Directory of open access journals\\
%TLA: Three letter acronym\\
%LD: linear dichroism}

%%%%%%%%%%%%%%%%%%%%%%%%%%%%%%%%%%%%%%%%%%
%% optional
\appendix
%\section{}
%The appendix is an optional section that can contain details and data supplemental to the main text. For example, explanations of experimental details that would disrupt the flow of the main text, but nonetheless remain crucial to understanding and reproducing the research shown; figures of replicates for experiments of which representative data is shown in the main text can be added here if brief, or as Supplementary data. Mathemtaical proofs of results not central to the paper can be added as an appendix.

%\section{}
%All appendix sections must be cited in the main text. In the appendixes, Figures, Tables, etc. should be labeled starting with `A', e.g., Figure A1, Figure A2, etc.

%%%%%%%%%%%%%%%%%%%%%%%%%%%%%%%%%%%%%%%%%%
\bibliographystyle{unsrt}

%=====================================
% References, variant A: internal bibliography
%=====================================
%\renewcommand\bibname{References}
%\begin{thebibliography}{999}
% Reference 1
%\bibitem{ref-journal}
%Lastname, F.; Author, T. The title of the cited article. {\em Journal Abbreviation} {\bf 2008}, {\em 10}, 142-149.
% Reference 2
%\bibitem{ref-book}
%Lastname, F.F.; Author, T. The title of the cited contribution. In {\em The Book Title}; Editor, F., Meditor, A., Eds.; Publishing House: City, Country, 2007; pp. 32-58.
%\end{thebibliography}

%=====================================
% References, variant B: external bibliography
%=====================================
\bibliography{refs_paper}

%%%%%%%%%%%%%%%%%%%%%%%%%%%%%%%%%%%%%%%%%%
%% optional
%\sampleavailability{Samples of the compounds ...... are available from the authors.}

%%%%%%%%%%%%%%%%%%%%%%%%%%%%%%%%%%%%%%%%%%
\end{document}